\documentclass[11pt,a4paper]{article}
\RequirePackage[OT1]{fontenc}
\RequirePackage{amsthm,amsmath}
\usepackage{amsfonts}
\numberwithin{equation}{section}
\theoremstyle{plain}
\newtheorem{theorem}{Theorem}[section]
\newtheorem{example}[theorem]{Example}
\newtheorem{definition}[theorem]{Definition}
\newtheorem{lemma}[theorem]{Lemma}
\newtheorem{notation}[theorem]{Notation}

\newtheorem{remark}[theorem]{Remark}
\usepackage{shuffle}
\usepackage{graphicx}
\usepackage{subfigure}

\usepackage{color}

\begin{document}
%\begin{frontmatter}
\title{Learning from the past, predicting the statistics for the future, learning an evolving system}
%\runtitle{Learning an evolving system}
%%\thankstext{T1}{Footnote to the title with the ``thankstext'' command.}

\author{Daniel Levin \and Terry Lyons \footnote{Mathematical Institute and Oxford-Man Institute, University of Oxford. Supported by ERC (Grant Agreement No.291244 Esig) and EPSRC (Grant Reference: EP/F029578/1).} \and Hao Ni\footnote{{Oxford-Man Institute of Quantitative Finance, University of Oxford. Supported by ERC (Grant Agreement No.291244 Esig). }} \\
\centerline{ University of Oxford} }
%\affiliation{University of Oxford}
 % University of Oxford 
  %\texttt{jbloggs@latex.wizard}
%\date{\today}

%\begin{aug}
%\author{\fnms{Daniel} \snm{Levin}\thanksref{t1}\ead[label=e1]{},}%{first@somewhere.com}},
%\author{\fnms{Terry} \snm{Lyons}\thanksref{t1,t2, t3} \ead[label=e2]{terry.lyons@oxford-man.ox.ac.uk}}
%\and
%\author{\fnms{Hao} \snm{Ni}\thanksref{t1, t2}\ead[label=e3]{ ni@maths.ox.ac.uk}}

%\thankstext{t1}{Supported by the Oxford-Man Institute of Quantitative Finance.}
%\thankstext{t2}{Supported by ERC (Grant Agreement No.291244 Esig). }
%\thankstext{t3}{Supported by  EPSRC (Grant Reference: EP/F029578/1).}
%\runauthor{Levin, Lyons and Ni.}

%\affiliation{University of Oxford}
%\affiliation{Oxford-Man Institute of Quantitative Finance\thanksmark{m2}}
%\address{Oxford-Man Institute of Quantitative Finance\\
%Eagle House, Walton well Road\\
%Oxford, OX2 6ED\\
%\printead{e2}\\
%\printead{e3}}%\printead*{e2}
%\end{aug}

\maketitle
\begin{abstract}
Rough path theory is a mathematical toolbox providing for the deterministic modelling of interactions between highly oscillatory systems (rough paths). The theory is rich enough to capture and extend classical It{\^o} stochastic calculus but has far wider significance. Fundamental to this approach is the realisation that the evolving state of the system is best described or measured over short time intervals by considering the realised effect of the system on certain controlled systems (the measurement instruments). The approach is crucially more efficient than an approach based on sampling. Even in the limit, the sequence of time series obtained through finer sampling need not provide a sufficient statistic that could predict the effects of a Brownian motion in an interaction (\cite{friz2015physical})! This divergence in the analysis arising from the choice of feature set is relevant in any context where the data streams are highly oscillatory on normal time scales.
 
We bring the theory of rough paths to the study of non-parametric statistics on streamed data and particularly to the problem of regression where the input variable is a stream of information, and the dependent response is also (potentially) a path or a stream. We explain how a certain graded feature set of a stream, known in the rough path literature as the signature of the path, has a universality that allows one to characterise the functional relationship summarising the conditional distribution of the dependent response. At the same time this feature set allows explicit computational approaches through linear regression. 

Truncating \emph{the signature of a stream} leads to an efficient local description for a rough path. The intrinsic non-linearity provides a provable efficiency gain for the signature over linear descriptions of the path segment with the same number of features.  The gain offered by this approach is a potentially significant, and even transformational dimension reduction in the feature set needed to describe a highly oscillatory data stream.
 
By way of illustration, we consider prediction for stationary time series such as the AR model and ARCH models. In the numerical examples we examined, our more universal method for prediction achieves similar accuracy to the Gaussian Process (GP) approach with much lower computational cost especially when the sample size is large. We give several examples to show how, on normal time scales, a low dimensional statistic can be effective to predict the effects of a data stream.
\end{abstract}

%\begin{keyword}[class=MSC]
%\kwd[Primary ]{62G08}
%\kwd[; secondary ]{62M10}
%\end{keyword}

%\begin{keyword}
%\kwd{Nonparametric regression}
%\kwd{Rough path theory}
%\kwd{Time series analysis}
%\end{keyword}

%\end{frontmatter}
\section{Introduction}
\subsection{Motivation for signatures as a feature set}
The interpretation and management of streamed data has become an important area of research for the computer science, database, and statistics communities. We are interested in tackling the question of efficient description of complex data streams in terms of their potential to affect other systems with which they interact. 

Consider a data stream built at some fine scale out of time stamped values, but which on normal scales seems continuous.  To capture the effects of this stream using classical approaches may involve sampling at very high frequencies, or even collecting all the ticks. The high dimensionality of the description has potential to make the identification of features for regression very challenging. 

There is a fundamental issue when one accesses a stream through sampling. It is quite possible for two streams to have very different effects and yet have near identical values when sampled at very fine levels (\cite{friz2015physical}).  We can see from our experience in the numerical approximation of stochastic differential equations (SDEs) that sampling can be a provably ineffective way to capture informative data. In \cite{clark1980maximum} Clark and Cameron study the numerical approximation of a solution to a multidimensional SDE driven by Brownian motion and ask how well can one approximate the solution at time $1$ (an effect) if one has only sampled the driving signal (stream) at $n$ deterministic locations in $[0, 1]$. \footnote{It might be tempting to think that Fourier series or wavelets might avoid these issues better than sampling? Unfortunately, Clarke, Cameron and Dickinson rule them out as no more effective features than sampling because they are also linear functionals of the path. (\cite{clark1980maximum}) }Clark and Cameron establish a lower bound demonstrating that in general sampling is relatively ineffective as a way of acquiring useful predictive information about the Brownian motion. Dickinson,  (Proposition 17, \cite{dickinson2007optimal}), shows that any $n$ dimensional \emph{linear} functional of the stream is a similarly poor predictor. 

On the other hand there are efficient $n$-dimensional descriptions based on the method we are talking about later and effective high order predictors. They allow orders of magnitude better predictions with a given dimension of description or feature set. Rough path theory exploits this approach in a decisive way for SDEs. In doing so it provides useful feature sets for capturing streams in much broader contexts. 

In this paper we will combine the essentially nonlinear theory of rough paths and regression techniques to answer the question how one should model and forecast the effects of the data streams robustly and effectively. We provide some examples for comparison. The practical tests show similar accuracy for our methods and Gaussian Process methods;  we are not surprised by this because at the very fine sampled scale our examples fits Gaussian Processes framework; however the performance of our approach is many factors better in terms of computational time.

One special case is the autoregressive model of discrete data streams with exogenous inputs (ARX or NARX) models,  and there has been a substantial body of literature to treat it either parametrically or non-parametrically (e.g. \cite{billings2005new}, \cite{bollerslev1990modelling}, \cite{hardle1990applied}). The signature-based approach can be applied to this setting and numerical results support that it is favorable in terms of comparable accuracy but lower computational cost when benchmarking with the non-parametric GP method.

The methods used to develop and justify this approach, such as the signature of a path and the shuffle product of tensors, are standard tools in the theory of rough paths and seem appropriate in this context of regression as well and provide a surprisingly unified and non-parametric approach. 

\subsection{Signatures} 
Summarising a stream over intervals is not a new idea. A function is said to be differentiable precisely when the chord is a good local summary of the function over small intervals. In this case the effect of smooth path can be very well approximated by its piecewise chordal approximation. The signature allows one to extend this idea in a fundamental way, and provide an effective way of describing the complex oscillatory streams as well as the smooth ones. 

Once one has fixed attention on a segment of the stream, then the chord is independent of how one parameterises the path in the interval. The same is true of the signature. If one were to split the segment into smaller segments and rearrange them, the chord would not notice this either. It is a commutative summary of the stream on the segment. The higher terms in the signature are not like this at all. They become progressively less commutative with degree. This explains their efficiency as they capture order of events.  

Mathematically the signature is a homomorphism from the monoid of streams under concatenation into the grouplike elements of a closed tensor algebra! We can truncate it and these truncations provide a graded summary, and progressively more order sensitive of the data streams in terms of its effects. Rough path theory teaches us to exploit the total ordering of the data and summarise the data over segments. However in this paper we only consider the signature over the whole interval. The signature (Definition \ref{DefSig}) provides a way to talk holistically about information in the segments of the stream. It does so without directly engaging with the micro-structure of the data, the individual ticks etc.

\subsection{The expected signature - a model-free description of empirical measure on paths}
If the signature is an effective way of describing a data stream, then a second and important remark is that the expected signature of a random stream provides a non-parametric way to fully describe the measure on the stream space (see Theorem \ref{theoremES}). We think that the ability to systematically characterize the empirical probability measure on the streams is exciting. The grading allows valuable approximative description for the law of the data stream.

\subsection{Literature Review on Rough Paths Theory in Statistical Inference}
Consider a controlled differential equation in the form of (\ref{controlledEquation}), where $Y$ represents the state of a complex system, and it evolves as a function $f$ of its present state and the infinitesimal increments of a driving signal $X$.
\begin{eqnarray}\label{controlledEquation}
dY_{t} = g(Y_{t})  dX_{t}, Y_{0} = y_{0}.
\end{eqnarray}
The theory of rough path was initially developed to capture and make precise the interactions between highly oscillatory and non-linear systems. It allows a deterministic treatment of SDEs or even controlled differential equation driven by much rougher signals than semi-martingales. The signature of a path is a fundamental object in rough path theory. The signature of a path $X$ arises naturally when solving a linear differential equation driven by $X$ and it is also a basis to represent a solution to a general controlled differential equation (\ref{controlledEquation}) with smooth vector fields as an analogy to Taylor expansion. 

So far the rough path theory has been developed and used mainly in the field of stochastic analysis (\cite{lyons1998differential}). Recently several papers concerned with applying the rough paths theory to statistical inference have appeared. For example, in \cite{papavasiliou2011parameter} the expected signature matching approach has been proposed for parameter estimation of stochastic processes as an extension to the moment matching method. To the best of our knowledge, the early version of our paper has been the first paper to introduce the signature of a path as a feature set for regression analysis; our focus has been on building a systematic framework for the regression problem with streamed inputs using the concept of the signature, while the other papers focus more on developing examples of application. The signature-based approach has been successfully applied to the classification and prediction problems of financial time series  (see \cite{Greg} , \cite{lyons2014feature}).  

The two main objectives of this paper are to present the elements of rough paths theory to a statistical audience, and to demonstrate its usefulness for statistical inference for streamed data.  
\subsection{The outline of the paper}
We start with introduction of the key concept - \emph{the signature of a path} and its relevant properties which we will need in the sequel.  In Section \ref{SignatureApproachSection}, a general signature approach is proposed in order to estimate the  functional relationship between an input stream and the corresponding noisy output, and its calibration and forecast methods are discussed. 

In Section \ref{timeSeriesEs}, we apply our method to stationary time series modeling and explore the link between our models and classic parametric time series models. Finally in Section \ref{NumericalExamples}, we examine three examples using our method, in which three sets of time series data are generated based on either linear or non-linear time series models and benchmark our approach with the standard AR calibration as well as non-parametric GP approach.
\section{Preliminaries}\label{signatureSection}
In this section, we review the basics of the theory of rough paths, which are useful for statistical inference. We aim to provide a self-supporting but minimalist account of key concepts and results of rough path theory. For readers interested in more details about the theory of rough paths,  we refer to \cite{friz2014course}, \cite{RoughPaths} and \cite{lyons2002system}. 
\subsection{The signature of a path}
Let $X: J \rightarrow E:= \mathbb{R}^{d}$ be a $d$-dimensional path where $J$ is a compact interval. We say that a path $X$ is of finite $p$-variation for certain $p \geq 1$ if the $p$-variation of $X$ is defined by
\begin{eqnarray*}
\vert\vert X \vert\vert_{p, J} = \left( \sup_{\mathcal{D}_{J} \subset J}  \sum_{l} \vert\vert X_{t_{l}} - X_{t_{l-1}} \vert\vert^{p} \right)^{1/p} < \infty,
\end{eqnarray*}
where the supremum is taken over all possible finite partitions $\mathcal{D}_{J} = \{t_{l}\}_{l}$ of the interval $J$. Let $\mathcal{V}^{p}(J, E)$ denote the set of  any continuous path $X: J \rightarrow E$ of finite $p$-variation.
\begin{example}
A Brownian motion path up to time $T$, denoted by $B_{[0, T]}$ has finite $p$-variation almost surly if and only if $p >2$. 
\end{example}
The signature of a path $X$ in $E$ is a formal power series with $d$ non-commutative indeterminates whose coefficients are iterated integrals of the path. The definition of the signature of a path is given as follows.
\begin{definition}[The Signature of a Path]\label{DefSig}
Let $J$ be a compact interval and $X \in \mathcal{V}^{p}(J, E)$ such that the following integration makes sense. The signature $S(X)$ of $X$ over the time interval $J$ is defined as follows
\begin{eqnarray*}
\mathbf{X}_{J} = (1, X^{1}, \cdots, X^{n}, \cdots),
\end{eqnarray*}
where for each integer $n \geq 1$,
\begin{equation*}
X^{n}=\underset{\underset{u_{1},...,u_{n}\in J}{u_{1}<...<u_{n}}}{\int \cdots \int } dX_{u_{1}}\otimes ...\otimes dX_{u_{n}} \in E^{\otimes n}.
\end{equation*}
The truncated signature of $X$ of order $n$ is denoted by $S^{n}(X)$, i.e. $S^{n}(X) = (1,X^{1},...,X^{n})$, for every integer $n \geq 1$.
\end{definition}
\begin{example} 
\begin{enumerate}
\item When $X$ is of finite 1-variation, the integration can be understood in the sense of the Stieltjes integral;
\item When $X$ is a Brownian motion path, its signature can be defined in the sense of the It\^o integral or the Stratonovich integral.
\end{enumerate}
\end{example}
\begin{notation}
Let $T((E))$ denote the tensor algebra space
\begin{eqnarray*}
T((E)): = \left \{ (a_{0}, a_{1}, \cdots, a_{n}, \cdots) \vert \forall n \geq 0, a_{n} \in E^{\otimes n} \right \}. 
\end{eqnarray*}
Let $T^{n}(E)$ denote the $n^{th}$ truncated tensor algebra space
\begin{eqnarray*}
T^{n}(E): = \bigoplus_{i = 0}^{n} E^{\otimes i},
\end{eqnarray*}
where by convention $E^{\otimes 0} = \mathbb{R}$.
\end{notation}
Obviously the signature of $X$ is an element of $T((E))$, and the truncated signature of $X$ of order $n$ lies in $T^{n}(E)$. To define the coordinate signature of $X$, let us introduce the notation $A^{*}$, which denotes the set of multi-indexes with entries in $\{1, \cdots, d\}$. The length of an index $I$ is denoted by $\vert I\vert$. For any index $I \in A^{*}$ whose length is $n$, $I$ can be written as $I  = (i_{1}, \dots, i_{n})$, where $i_{j} \in \{1, \cdots, d\}, \forall j \in \{1, \cdots, n\}$. By convention the empty index denoted by $() $, and it is in $A^{*}$. 
\begin{remark}
Let $(e_{i})_{i=1}^{d}$ be a basis of $E$. For any positive integer $n$, the space $E^{\otimes n}$ is isomorphic to the free vector space generated by all the words of length $n$ in $A^{*}$ and $(e_{i_{1}} \otimes \dots \otimes e_{i_{n}})_{(i_{1}, \cdots, i_{n}) \in \{1, \cdots, d\}^{n}}$ form a basis of $E^{\otimes n}$. The coordinate signature of $X$ indexed by $I$ denoted by $C_{I}(X)$ is defined to be
\begin{eqnarray*}
C_{I}(X)  = \underset{\underset{u_{1},...,u_{n}\in J}{u_{1}<...<u_{n}}}{\int \cdots \int} dX_{u_{1}}^{(i_{1})}dX_{u_{2}}^{(i_{2})} \ldots dX_{u_{n}}^{(i_{n})} \in \mathbb{R}.
\end{eqnarray*}
Thus it follows that
\begin{eqnarray*}
S(X) = 1 + \sum_{n=1}^{\infty} \sum_{\vert I \vert = n} C_{I}(X)  e_{i_{1}} \otimes e_{i_{2}}\cdots \otimes e_{i_{n}} \in T((E)),
\end{eqnarray*}
where the second sum is taken over all indexes with length $n$.
\end{remark}
\begin{example}\label{EgSigOfLinearPath}
Let $X: [0, T] \rightarrow E$ be a \textbf{linear} path. It implies that for any $t \in [0, T]$, $X_{t} = X_{0} + \frac{X_{T}- X_{0}}{T}t$. For any integer $n \geq 1$ and any $I = (i_{1}, \dots, i_{n})$,
\begin{eqnarray*}
C_{I}(X) &=&  \underset{\underset{u_{1},...,u_{n}\in J}{u_{1}<...<u_{n}}}{\int \cdots \int} d \left(\frac{(X_{T}- X_{0})^{(i_{1})}u_{1}}{T}\right)\ldots d \left(\frac{(X_{T}- X_{0})^{(i_{n})}u_{n}}{T}\right) \\
&=& \prod_{j = 1}^{n} \frac{ (X_{T}^{(i_{j})} - X_{0}^{(i_{j})})}{T^{n}}  \underset{0< u_{1}<...<u_{n} <T}{\int \cdots \int} du_{1}\cdots du_{n}
\end{eqnarray*}
By a simple calculation, it follows that 
\begin{eqnarray*} 
\underset{0< u_{1}<...<u_{n} <T}{\int \cdots \int} du_{1}\cdots du_{n} = \frac{T^{n}}{n!},
\end{eqnarray*} 
and $C_{I}(X)$ can be further simplified to
\begin{eqnarray*}
 \frac{1}{n!}\prod_{j = 1}^{n} (X_{T}^{(i_{j})} - X_{0}^{(i_{j})}).
\end{eqnarray*}
It is convenient to rewrite it in the tensor product form as follows:
\begin{eqnarray*}
S(X) = \exp( X_{T} - X_{0}) , %X_{t}^{(i)} - X_{s}^{(i)})e_{i}
\end{eqnarray*}
where for any $a \in T((E))$, 
\begin{eqnarray*}
\exp(a) := \sum_{n \geq 0} \frac{1}{n!}a^{\otimes n}.
\end{eqnarray*}
\end{example}
The concept of a rough path is a generalization of the signature of a path and it is defined in a way to capture the analytic and algebraic properties of the signature of path. A geometric $p$-rough path is a $p$-rough path which can be represented as a limit of $1$-rough path in $p$-variation path topology, and the space of geometric $p$-rough paths is denoted by  $G\Omega_{p}(E)$.  The definition of rough paths and geometric rough paths can be found in Appendix \ref{AppendixRoughPath}. For simplicity in most of the paper, we use the signature of a path.
\subsection{The origin of signatures}
The signature of a path was first studied by a geometer K.T Chen (\cite{chen1957integration}, \cite{chen1977iterated}). Iterated integrals of continuous smooth paths naturally arise when applying Picard iteration to solve controlled ordinary differential equations in form of (\ref{controlledEquation}). More specifically, let us consider the case where the driving path $X \in \mathcal{V}^{1}([0, T], E)$, the procedure of Picard iteration is to define a sequence of $\{Y(n): [0, T] \rightarrow W:=\mathbb{R}^{\tilde{d}}\}_{n \geq 1}$ recursively as follows for every $t \in [0, T]$, 
\begin{eqnarray*}
Y(0)_{t} &\equiv& y_{0};\\
Y(n+1)_{t} &=& y_{0} + \int_{0}^{t} g(Y(n)_{s}) dX_{s}.
\end{eqnarray*} 
By the fixed point argument, $Y(n)$ converges in 1-variation norm and the limit of $Y(n)$ is a solution to (\ref{controlledEquation}) under some appropriate condition. In Example \ref{LinearDiffEqnexample} of linear differential equations, by the $n^{th}$ Picard iteration $Y_{n}$ is simply a linear functional on the signature of the driving path of order $n$, and so is the limit of $Y_{n}$, which motivates to propose the concept of the signature of a path. 
\begin{example}[Linear Controlled Differential Equation]\label{LinearDiffEqnexample} Suppose $X \in \mathcal{V}^{1}([0, T], E)$ and $g:W \rightarrow L(E, W))$ is a bounded \textbf{linear} map. Then it follows that
\begin{eqnarray*}
Y(n)_{t} &=& \left(I + \sum_{k = 1}^{n} g^{\otimes k}  \underset{\underset{u_{1},...,u_{n}\in J}{u_{1}<...<u_{n}}}{\int \cdots \int } dX_{u_{1}}\otimes ...\otimes dX_{u_{n}}\right)y_{0},
\end{eqnarray*}
and $Y(n)_{t}$ converges to $Y_{t}$ as $n$ tends to infinity, and $Y_{t}$ yields the following representation 
\begin{eqnarray*}
 Y_{t} &=&\left (I + \sum_{k = 1}^{\infty} g^{\otimes k}  \underset{\underset{u_{1},...,u_{n}\in J}{u_{1}<...<u_{n}}}{\int \cdots \int } dX_{u_{1}}\otimes ...\otimes dX_{u_{n}}\right)y_{0},
\end{eqnarray*}
where $g^{\otimes k}(x_{1} \otimes \cdots \otimes x_{k}) = g(x_{1}) \cdots g(x_{k})$.
\end{example} 
When $g$ is a smooth vector field, we have an extension of the classical Taylor expansion on the path space.  Let us define $g^{\circ n}: W \rightarrow L(E^{\otimes n}, W)$ inductively by
\begin{eqnarray*}
g^{\circ 1} &=& g;\\
g^{\circ n+1} &=& D(g^{\circ n})g.
\end{eqnarray*}
If one regards the signature of a path as monomials on the paths space, then it is natural to define step-$N$ Taylor expansion for $Y_{t}$ by $\tilde{Y}(N)_{t}$ as follows:
\begin{eqnarray*}
\tilde{Y}(N)_{t} = y_{0} + \sum_{n= 1}^{N} g^{\circ n}(y_{0}) \underset{\underset{u_{1},...,u_{n}\in J}{u_{1}<...<u_{n}}}{\int \cdots \int } dX_{u_{1}}\otimes ...\otimes dX_{u_{n}}.
%Y_{t} &=& y_{0} + \sum_{i = n}^{\infty} g^{\circ n}(y_{0}) \underset{\underset{u_{1},...,u_{n}\in J}{u_{1}<...<u_{n}}}{\int \cdots \int } dX_{u_{1}}\otimes ...\otimes dX_{u_{n}}.
\end{eqnarray*}
It is noted that $\tilde{Y}(N)_{t}$ is linear in the truncated signature of $X$ of order $N$. Moreover the error bounds of $\tilde{Y}(N)_{t}$ to approximate $Y_{t}$ yields a factorial decay in terms of $N$, which is summarised in the following theorem (\cite{friz2008euler}).  
\begin{theorem}\label{FactorialDecayTheorem}%[\cite{friz2008euler}]
Let $Y: [0, T] \rightarrow W$ be a continuous path and satisfy (\ref{controlledEquation}), where $X: [0, T] \rightarrow E$ is a path of finite 1-variation and $g: W \rightarrow L(E, W)$ is a smooth vector field. Then for any integer $N >0$, 
\begin{eqnarray*}
\vert Y_{t} - \tilde{Y}(N)_{t}\vert \leq \vert \vert g^{\circ N} \vert \vert_{\infty} \vert\vert Dg \vert\vert_{\infty} \frac{\vert X \vert_{1, [0, t]}^{N+1}}{N!}.
\end{eqnarray*} 
\end{theorem}
This result has been extended to the case where the driving path is any $p$-geometric rough path, and $g$ is a $Lip(\gamma)$ vector field where $\gamma > p-1$ (Theorem 1, \cite{boedihardjo2015uniform}). 
\subsection{The properties of signatures}
In this subsection we summarize the main properties of signatures, which are important for statistical inference. The first property is Chen's identity (Theorem \ref{ChenIdentity}), which asserts the multiplicative property of the signature of a path.
\begin{definition}
Let $X:[0,s]\longrightarrow \emph{E}$ and $Y:[s,t]\longrightarrow \emph{E}$
be two continuous paths. Their concatenation is the path $X\ast Y$ defined
by
\begin{equation*}
(X\ast Y)_{u}=\left\{
\begin{array}{ll}
X_{u}, & u\in \left[ 0,s\right]; \\
X_{s}+Y_{u}-Y_{s}, & u\in \left[ s,t\right],%
\end{array}%
\right.
\end{equation*}
where $0 \leq s \leq t$.
\end{definition}
\begin{theorem}[Chen's Identity]\label{ChenIdentity}
Fix $p \in [1, 2)$ and $ 0 \leq s \leq t$. Let $X \in \mathcal{V}^{p}([0,s], \emph{E})$ and $Y \in \mathcal{V}^{p}([s, t], \emph{E})$. Then it holds that
\begin{equation}
S(X\ast Y)=S(X)\otimes S(Y).
\end{equation}
\end{theorem}
The proof can be found in \cite{RoughPaths}.
\begin{example}
Let $X$ be a $E$-valued piecewise linear path, i.e. $X$ is the concatenation of a finite number of linear paths, and in other words there exists a positive integer $l$ and linear paths $X_{1}, X_{2}, \dots, X_{l}$ such that
\begin{eqnarray*}
X = X_{1} * X_{2} * \dots * X_{l}.
\end{eqnarray*}
Recall the analytic formula for the signature of a linear path (see Example \ref{EgSigOfLinearPath}). Chen's identity provides a method to compute the signature of a piecewise linear path, as 
\begin{eqnarray*}
S(X) = \overset{l}{\underset{i = 1}{\bigotimes}} \exp(X_{i}).
\end{eqnarray*}

\end{example}
The second property of the signature is invariance under time re-parameterizations (Lemma 1.6, \cite{RoughPaths}).
\begin{lemma}\label{SigTimeInvariance}
Let $X \in \mathcal{V}^{1}( [0, T], E)$ and $\lambda: [0, T] \rightarrow [T_{1}, T_{2}]$ be a non-decreasing surjection and define $X_{t}^{\lambda} := X_{\lambda_{t}}$ for the reparamterization of $X$ under $\lambda$. Then for every $s, t \in [0, T]$,
\begin{eqnarray*}
S(X)_{\lambda_{s}, \lambda_{t}} = S(X^{\lambda})_{s,t}.
\end{eqnarray*}
\end{lemma}
Another important property of the signature is the uniqueness of the signature of a path. It was first proved for any continuous path of finite 1-variation (\cite{UniquenessOfSignature}), which has been extended to the case of a geometric $p$-rough path for some $p >1$ (\cite{boedihardjo2014signature}).
\begin{theorem}[Uniqueness of signature]\label{UniquenessSigThm}
Let $X \in \mathcal{V}^{1}( [0, T], E)$. Then $S(X)$ determines $X$ up to the tree-like equivalence. (\cite{UniquenessOfSignature})
\end{theorem}
Heuristically,  a tree-like path can be thought as being a null-path as a control; the trajectories are completely canceled out by themselves. The precise definition of tree-like equivalence can be found in \cite{UniquenessOfSignature}. The following lemma is a useful sufficient condition for the uniqueness of the signature.
\begin{lemma}
Let $X \in \mathcal{V}^{1}( [0, T], E)$ with a fixed starting point and at least one coordinate of $X$ is a monotone function. Then $S(X)$ determines $X$ uniquely.
\end{lemma}
Furthermore, a sufficient condition under which a general class of multi-dimensional stochastic processes over a fixed time interval is determined by their signatures up to re-parameterization with probability one has been established in \cite{boedihardjo2014uniqueness} and it is an extension of the original work on the Brownian path (\cite{le2013stratonovich}). For simplicity we state the result only for the Stratonovich signatures of Brownian motion.
\begin{theorem}[Uniqueness of signature of Brownian motion]\label{UniquessBMThm}
Let $B$ denote a standard $d$-dimensional Brownian motion in $E$ and $S(B_{[0, T]})$ denote the Stratonovich signatures of $B$ up to time $T$, where $T >0$.  Then all Brownian motion sample paths up to time $T$ are determined by their signature $S(B_{[0, T]})$ up to time re-parameterization almost surely.
\end{theorem}
\subsection{Linear forms on signatures}
Recall that $\{e_{i}\}_{i = 1}^{d}$ is a basis of $E$. Then let $(e_{i}^{*})_{i=1}^{d}$ be a basis of the corresponding dual space $E^{*}$. Thus for every $n \in \mathbb{N}$, $(e_{i_{1}}^{*} \otimes \cdots \otimes  e_{i_{n}}^{*})$ can be naturally extended to $(E^{*})^{\otimes n}$ by identifying the basis $\left( e_{I}^{*} = e_{i_{1}}^{*} \otimes \dots \otimes e_{i_{n}}^{*} \right)$ as
\begin{eqnarray*}
\langle e_{i_{1}}^{*}\otimes \dots \otimes e_{i_{n}}^{*}, e_{j_{1}}\otimes \dots \otimes e_{j_{n}}\rangle = \delta_{i_{1}, j_{1}}\dots \delta_{i_{n}, j_{n}}.
\end{eqnarray*}
The linear action of $(E^{*})^{\otimes n}$ on $E^{\otimes n}$ extends naturally to a linear mapping $(E^{*})^{\otimes n} \rightarrow T((E))^{*}$ defined by
\begin{eqnarray*}
e_{I}^{*}(\mathbf{a})= e_{I}^{*}(a_{n}),
\end{eqnarray*}
for every $I = (i_{1}, \dots, i_{n}) \in A^{*}$. \\
The linear forms $\{e_{I}^{*}\}_{I \in A^{*}}$ form a basis of $T(E^{*})$. Let $T((E))^{*}$ denote the space of linear forms on $T((E))$ induced by $T(E^{*})$. Let $S(\mathcal{V}^{p}([0, T], E))$ denote the range of the signature of  a path $X \in \mathcal{V}^{p}( [0, T], E)$. For any index $I \in A^{*}$, define $\pi^{I}$ as $e_{I}^{*}$ restricting the domain $S(\mathcal{V}^{p}([0, T], E))$, in formula
\begin{eqnarray*}
\pi^{I}(S(X)) = e_{I}^{*}(S(X)).
\end{eqnarray*}
For any $I, J \in A^{*}$, the pointwise product of two linear forms $\pi^{I}$ and $\pi^{J}$ as real valued functions is a quadratic form on $S(\mathcal{V}^{p}([0, T], E))$, but it is remarkable that it is still a linear form (Theorem \ref{shuffle_theorem}). Let us introduce the definition of the shuffle product.
\begin{definition}
We define the set $S_{m,n}$ of $(m, n)$ shuffles to be the subset of permutation in the symmetric group $S_{m+n}$ defined by
\begin{eqnarray*}
S_{m,n} = \{ \sigma \in S_{m+n}: \sigma(1) < \dots < \sigma(m), \sigma(m+1) < \dots < \sigma(m+n)\}.
\end{eqnarray*}
\end{definition}
\begin{definition}
The shuffle product of $\pi^{I}$ and $\pi^{J}$ denoted by $\pi^{I} \shuffle \pi^{J}$ is defined as follows:
\begin{eqnarray*}
\pi^{I} \shuffle \pi^{J}  = \sum_{\sigma \in S_{m, n}} \pi^{(k_{\sigma^{-1}(1)}, \dots, k_{\sigma^{-1}(m+n)})},
\end{eqnarray*}
where $I = (i_{1}, i_{2}, \cdots, i_{n}), J =( j_{1},j_{2},\cdots, j_{m})$ and $(k_{1}, \dots, k_{m+n}) =( i_{1}, \cdots, i_{n}, j_{1},$\\
$\cdots, j_{m})$.
\end{definition}
\begin{theorem}[Shuffle Product Property]\label{shuffle_theorem}
Let $p \in [1, 2)$. For any $I_{1}, I_{2} \in A^{*}$, it holds that for any path $X \in \mathcal{V}^{p}(J, E)$,
\begin{eqnarray*}
\pi^{I_{1}}( S(X)) \pi^{I_{2}}(S(X)) = (\pi^{I_{1}} \shuffle \pi^{I_{2}})(S(X)).
\end{eqnarray*} 
\end{theorem}
For intuitive explanation of key properties of signatures related to statistical inference we refer to Section 2.2 in \cite{Greg}.
\subsection{Expected signature of stochastic processes}
\begin{definition}
Given a probability space $(\Omega, P, \mathcal{F})$, $X$ is a $E$-valued stochastic process. Suppose that for every $\omega \in \Omega$, the signature of $X(\omega)$ denoted by $\mathbf{X}(\omega)$ \emph{(}or $S(X(\omega))$\emph{)} is well defined a.s and under the probability measure $P$, its expectation denoted by $\mathbb{E}[\mathbf{X}(\omega)]$ is finite. We call $\mathbb{E}[\mathbf{X}(\omega)]$ the expected signature of $X$.
\end{definition}
As the signature of a path can be thought as non-commutative monomials on the path space, the expected signature of a random path plays a similar role as that the moment generating function of a random variable does. It was firstly proved that the law of the signature with compact support is determined by its expected signature in \cite{Fawcett}. Recently this result has been strongly extended in \cite{chevyrev2013set}, by introducing a characteristic function of a random signature $\mathbf{X}$ defined by $\Phi_{\mathbf{X}}(M) = \mathbb{E}[M(\mathbf{X})]$, where $M$ is a unitary representation arising from a linear map  $M$ mapping $E$ into a unitary Lie algebra as an analogy of $\lambda$ in the characteristic function of a random variable $\lambda \mapsto \mathbb{E}[e^{\lambda X}]$ . The main result of \cite{chevyrev2013set} asserts that $\Phi_{\mathbf{X}}$ determines the measure on $\mathbf{X}$.  Here we give a sufficient condition for that the expected signature determines the law on signatures in Theorem \ref{theoremES}. Recall that $G\Omega_{p}(E)$ is the space of $p$-geometric rough paths.
\begin{theorem} [Proposition 6.1, \cite{chevyrev2013set}]\label{theoremES}
Let $\mathbf{X}$ and $\hat{\mathbf{X}}$ be two random variables taking values in $G\Omega_{p}(E)$ for some $p \geq 1$ such that $\mathbb{E}[\mathbf{X}] = \mathbb{E}[\hat{\mathbf{X}}]$ and $\mathbb{E}[\mathbf{X}]$ has an infinite convergence radius. Then $\mathbf{X} \overset{\mathcal{D}}{=} \hat{\mathbf{X}}$.
\end{theorem}
%%%%%%%%%%%%%%%%%%%%%%%%%%
\section{A general expected signature framework}\label{SignatureApproachSection}
\subsection{The expected signature model}\label{GeneralExpectedSignature}
%\subsection{Modeling data streams through regression}
In this subsection, we model the effects of data streams through the regression framework and set up a general regression problem on data streams. Let $V$ and $W$ be two Banach spaces and $J$ be a compact time interval. A $E$-valued data stream in time interval $J$ can be described as a function $X: \mathcal{D} \rightarrow E$, where $\mathcal{D}$ is the set of event times and $\mathcal{D} \subset J$. To cope with different time stamps, we embed $X$ to a function mapping from $J$ to $E$. One way to do this is to extend it to be piecewise linear, but we will see later that there are more intelligent approaches. We therefore assume the embedded function is continuous, even piecewise smooth on a very fine scale, but potentially very oscillatory and barely accessible at this scale. As we will explain later, the fine structure gives the meaning to iterated integrals, and differential equations etc. But it is not a convenient or efficient analytic framework to measure the effects of streams. The theory of rough paths uses a form of $p$-variation to complete this space allowing integrals, differential equations to be well defined on the completion. Recall that $\mathcal{V}^{p}(J, E)$ is the space of continuous functions mapping from $J$ to $E$ with finite $p$-variation. In the following discussion, we consider an element in $\mathcal{V}^{p}(J, E)$ as a data stream (or a path).  

We can formulate the effects of data streams as a dependent variable of a regression problem with an explanatory variable being a path in $\mathcal{V}^{p}(J, E)$. Potentially the dependent variable is a stream as well. Suppose  that we have observations of the input-output pair $\{ X_{i}, Y_{i} \}_{i =1}^{N}$, and the observations are assumed to satisfy 
\begin{eqnarray}\label{RegressionEquation}
Y_{i} = f(X_{i}) + \varepsilon_{i}, \forall i \in 1, \dots, n;
\end{eqnarray}
where $X_{i} \in \mathcal{V}^{p}(J, E)$, $Y_{i} \in \mathcal{V}^{p}([0,T], W) $, $\mathbb{E}[\varepsilon_{i} \vert  X_{i}] = 0$ and $f$ is a unknown function on the path space. In this paper, we aim to provide some new insights to answer the question about how to estimate the function $f$ accurately and effectively in this setting. As before we denote the signatures of $X$ and $Y$ by $\mathbf{X}$ and $\mathbf{Y}$ respectively.  Following classical techniques of non-parametric regression for the finite dimensional case, it is crucial to identify specific feature sets of the observed input (output) data to $\mathbf{linearize}$ the functional relationship between them. In this paper we propose to apply linear regression on \emph{the signature feature} of a path to solve this problem and we will demonstrate the advantages of this approach to tackle the general regression problem on the path space.  
%The major difficulty of this problem comes from the functional type of dependent variables. Learning from classical non-linear regression techniques for a finite dimensional case, the core idea is to transform the vector of inputs to its features such as to linearize the functional relationship between the output and the derived input features. Following  this idea of a linear basis expansion in input variable, we propose to use \emph{the signature of a path} as a feature set for paths. Formally  it reduces the non-linear regression problem to a linear one. Intuitively the signature of path can be thought as a generating function for the non-commutative polynomials on the path space and serves as a natural linear "basis" for smooth functions on streams that can be used for statistical inference. 

\begin{theorem}[Signature Approximation]\label{SignatureApproximationThm}
Suppose $f: S_{1} \rightarrow \mathbb{R} $ is a continuous function where $S_{1}$ is a compact subset of $S(\mathcal{V}^{p}(J, E))$. Then for every $\varepsilon > 0$, there exists a linear functional $L \in T((E))^{*}$ such that for every $ a \in  S_{1}$, 
\begin{eqnarray*}
\vert f(a) - L(a)\vert  \leq \varepsilon.
\end{eqnarray*}
\end{theorem}
\begin{proof}
Let $\mathcal{L}(S_{1})$ denote a family of all linear functions in $T((E))^{*}$ restricted to $S_{1}$. By the shuffle product property of signatures (Theorem \ref{shuffle_theorem}), $\mathcal{L}(S_{1})$ is an algebra. Since the $0^{th}$ term of the signature is always 1, this algebra contains constant functions. Moreover it separates the points. (Details can be found in the proof of Corollary 2.16 in \cite{RoughPaths}.) By Stone-Weierstrass theorem, $\mathcal{L}(S_{1})$ is dense in the space of continuous functions on $S_{1}$. 
\end{proof}
\begin{remark}
A special case for Theorem \ref{SignatureApproximationThm} is that $S_{1}$ is the set of signatures of any finite number of sample paths.
\end{remark}
The goal is to learn the conditional distribution of $Y$ given the information $X$, which is $\mathbb{E}[\mathbf{Y}\vert \mathbf{X}]$ if written in the language of rough paths. The underlying reasons are two-fold:
\begin{enumerate}
\item The signature of a path of bounded variation uniquely determines the path up to the tree-like equivalence (Theorem \ref{UniquenessSigThm});
\item Under certain regularity condition, the expected signature of stochastic process determines the measure on the random signatures according to Theorem \ref{theoremES}.
\end{enumerate}
When $\mathbb{E}[\mathbf{Y}\vert \mathbf{X}]$ is a continuous function of $\mathbf{X}$, according to Theorem \ref{SignatureApproximationThm}, $\mathbb{E}[\mathbf{Y}\vert \mathbf{X}]$ can be well approximated by a linear function on $\mathbf{X}$ locally. By adding a perturbation of noise, it is natural to come up with the following expected signature model: 
\begin{definition}[Expected Signature Model]
Let $X$ and $Y$ be two stochastic processes taking values in $E$ and $W$ respectively. Suppose that the signatures of $X$ and $Y$ denoted by $\mathbf{X}$ and $\mathbf{Y}$ are well defined a.s. Assume that
\begin{eqnarray*}
\mathbf{Y} = L(\mathbf{X}) + \varepsilon,
\end{eqnarray*}
where $\mathbb{E}[\varepsilon \vert \mathbf{X}] = 0$ and $L$ is a linear functional mapping $T((E))$ to $T((W))$.
\end{definition}

\begin{notation}
Let $\mu_{\mathbf{X}}$ and $\Sigma^{2}_{\mathbf{X}}$ denote the conditional expectation $\mathbb{E}[\mathbf{Y}\vert \mathbf{X}]$ and covariance of  $\mathbf{Y}$  conditional on $\mathbf{X}$ respectively:
\begin{eqnarray}
\Sigma^{2}_{\mathbf{X}}: A^{*} \times A^{*} &\rightarrow&  \mathbb{R} ;\nonumber \\
(I,  J) &\mapsto& \mathrm{Cov}(\pi^{I}(\mathbf{Y}) , \pi^{J}(\mathbf{Y}) \vert \mathbf{X}),
\end{eqnarray}
where $I, J \in A^{*}$.
\end{notation}
The following lemma asserts that $\Sigma^{2}_{\mathbf{X}}$ is determined by $\mu_{\mathbf{X}}$.
\begin{lemma}\label{MeanDeterminesVar}
Let $\mu_{\mathbf{X}}$ and $\Sigma^{2}_{\mathbf{X}}$ be defined as before. Then for every $I, J \in A^{*}$,
\begin{eqnarray*}
\Sigma^{2}_{\mathbf{X}}(I, J) = (\pi^{I} \shuffle \pi^{J})(\mu_{\mathbf{X}}) - \pi^{I}(\mu_{\mathbf{X}})\pi^{J}(\mu_{\mathbf{X}}).
\end{eqnarray*}
\end{lemma}
\begin{proof}
For each $I, J \in A^{*}$, by the definition of the conditional covariance
\begin{eqnarray*}
\Sigma_{\mathbf{X}}^{2}(I, J) = \mathbb{E}[\pi^{I}(\mathbf{Y})\pi^{J}(\mathbf{Y}) \vert \mathbf{X} ]-  \mathbb{E}[\pi^{I}(\mathbf{Y}) \vert \mathbf{X}]\mathbb{E}[\pi^{J}(\mathbf{Y}) \vert \mathbf{X}].
\end{eqnarray*}
Due to the shuffle product property of the signature, we have
\begin{eqnarray*}
\mathbb{E}[\pi^{I}(\mathbf{Y})\pi^{J}(\mathbf{Y}) \vert  \mathbf{X}]= \mathbb{E}[(\pi^{I} \shuffle \pi^{J} )(\mathbf{Y})  \vert  \mathbf{X}] =  (\pi^{I} \shuffle \pi^{J})(\mathbb{E}[\mathbf{Y} \vert \mathbf{X}]),
\end{eqnarray*}
and we obtain the following equations immediately:
\begin{eqnarray*}
\Sigma^{2}_{\mathbf{X}}(I, J) = (\pi^{I} \shuffle \pi^{J})(\mu_{\mathbf{X}}) - \pi^{I}(\mu_{\mathbf{X}})\pi^{J}(\mu_{\mathbf{X}}).
\end{eqnarray*}
\end{proof}
\begin{definition}
Let $\mu \in T((E))$. We say that the quadratic form $\Sigma^{2}: A^{*} \times A^{*} \rightarrow \mathbb{R}$ is induced by $\mu$, if and only if for every $I, J \in A^{*}$
\begin{eqnarray*}
\Sigma^{2}(I, J) = (\pi^{I} \shuffle \pi^{J})(\mu) - \pi^{I}(\mu)\pi^{J}(\mu).
\end{eqnarray*}
\end{definition}
\begin{definition}
Let $\mu  \in T^{2n}((E))$. We say that the quadratic form $\Sigma^{2}: A^{*}_{n} \times A^{*}_{n} \rightarrow \mathbb{R}$, is induced by $\mu$ if and only if for every $I, J \in A^{*}_{n}$,
\begin{eqnarray*}
\Sigma^{2}(I, J) = (\pi^{I} \shuffle \pi^{J})(\mu) - \pi^{I}(\mu)\pi^{J}(\mu),
\end{eqnarray*}
where $A^{*}_{n}$ is the set of words over $A$ of the length no more than $n$ .
\end{definition}

\textbf{Calibration and prediction}: Under the expected signature model, given a large number of samples $\{\mathbf{X}_{i}, \mathbf{Y}_{i}\}_{i = 1}^{N}$,  estimating the expected truncated signature of $Y$ of the order $m$ on condition of $\mathbf{X}$, i.e. $\rho_{m}(\mathbb{E}[\mathbf{Y} \vert \mathbf{X}])$, or in other words, the linear functional $\rho_{m} \circ f$, turns out to be the standard linear regression problem; the coordinate iterated integrals of $\mathbf{Y}$ are multi-dimensional regressands while the coordinate iterated integrals of $\mathbf{X}$ are explanatory variables. In practice, we need to consider the truncated signature $\mathbf{X}$ of certain order instead of the full signature, since the number of explanatory variables should be finite. Since the regression is linear, there are many standard linear regression methods ready to use. We might use regularization or variable selection techniques such as LASSO or SVD to avoid the collinearity of the design matrix and the overfitting issue. We call this calibration method the ES approach, where ES stands for expected signature. 

To measure the goodness of fit for the model, we simply use the mean squared error of the residuals $\{a_{i}\}_{i = 1}^{N}$ as an indicator of the fitting performance, where
\begin{eqnarray*}
a_{i} = \mathbf{Y}_{i}- \hat{f}( \mathbf{X}_{i}), \text{ }\forall i =1, \dots , N.
\end{eqnarray*}
Alternatively we use $R^{2}$ or the adjusted-$R^{2}$ as an indicator of the fitting performance. 
\subsection{Signatures - A new feature set of a path}
In this subsection, we summarize the main reasons for choosing the signature as the feature set of a path for statistical inference and machine learning, and emphasize the situations when the signature-based approach are preferable over other methods.

Theorem \ref{UniquenessSigThm} and Theorem \ref{UniquessBMThm} highlight the relation between a path and its signature in both deterministic and probabilistic settings. Roughly speaking, there is almost one-to-one correspondence between paths and signatures, and thus a smooth function on the path space can be often viewed as another smooth function on the signature space. By Theorem \ref{SignatureApproximationThm}, any smooth function on the signature can be well approximated by a linear functional on the signature up to any arbitrary precision, and the linear forms on signatures indeed serve basis functions to represent any smooth function on signatures locally.

For the case that an output can be modeled as a solution to a controlled differential equation driven by an input path, Theorem \ref{FactorialDecayTheorem} guarantees that there exists a linear functional on truncated signature, which can approximate the output well as an analogy to Taylor's expansion. The corresponding error bound decays factorially fast as the degree of the truncated signature grows. It is also a positive evidence to support the fact that coordinate signatures of lower order carry most informative information in terms of the effects of the path. In the next subsection, the numerical example will show how the signature features facilitates the estimation of the solution to a SDE without knowing the underlying dynamics and it leads to significant dimension reduction.

When time reparametization of an input path does not affect its output, the signature of this path can be used to summarize the path information effectively,  because the signature is invariant up to time re-parameterization (see Lemma \ref{SigTimeInvariance}) .

In practice, rather than a continuous path $X$, we usually observe a sequence of data points at discrete time stamps, but potentially it can be sampled on a very fine time scale. Recording a data stream tick by tick might not be an effective way to summarize the data streams in terms of its effects.  Consider an example of the evidence of insider trading. Which happened first - the trade or the news? Then it may not be enough to summarise the data on a daily, or even one minute basis to establish the answer. The time series structure leads one to consider large amounts of uninformative data for what were really just two events. The signature captures this information in the sign of the first two level of signature however it is initially sampled (see Example \ref{IntroExample}).
\begin{example}\label{IntroExample}
Suppose that we have price data series on a minutely basis and the time of news for one day. Suppose that there are only one change in price and  one news for this day. Let $X_{i}$ denote the price at the $i^{th}$ minute, and $Y_{i}$ is an indicator function of the news defined as follows:
\begin{eqnarray*}
Y_{i} = \begin{cases}0, &\mbox{if } i  < \tau; \\ 
 1, & \mbox{if } i \geq \tau. \end{cases} 
\end{eqnarray*}
where $\tau$ is the arrival time of the news. \\
Figure \ref{NewsThenPrice} and Figure \ref{PriceThenNews} plot the possible shapes of the path $(X_{i}, Y_{i})_{i = 0}^{T}$ when the news happened before or after the change of price respectively. The sign of 
\begin{eqnarray*} 
(\pi^{(1)}(S((X_{i}, Y_{i})_{i = 0}^{T}), (\pi^{(1, 2)} - \pi^{(2, 1)})(S((X_{i}, Y_{i})_{i = 0}^{T}) ) 
\end{eqnarray*} 
are discriminative for those two cases. It is much more effective than keeping all the minutely data especially when we do not know $\tau$ and the time of price change in advance.
\end{example}
\begin{figure}  
\caption{The case for that the news happens earlier than the price move}\label{NewsThenPrice}
  \centering
  \includegraphics[width= 4.0 in]{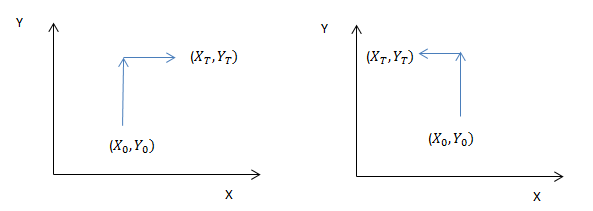}
\end{figure}
\begin{figure}  
\caption{The case for that the news happens later than the price move}\label{PriceThenNews}
  \centering
  \includegraphics[width= 4.0 in]{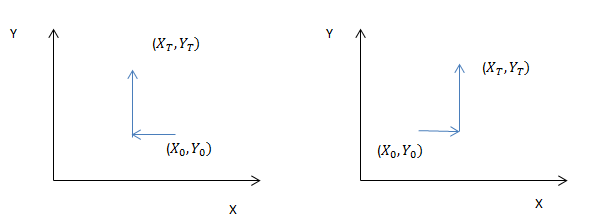}
\end{figure}

As it shows in Example \ref{IntroExample}, the discrete approximation of a data stream via sampling may not be an efficient representation of this stream. Neither is any linear functional of the data streams, such as Fourier series or wavelets. It is important to find effective non-linear feature sets of the data streams for this case. However, the standard non-linear feature sets of discrete time series, such as polynomial basis, which would contain much redundant information and leads to the curse of dimensionality. To the contrast, a discrete time series of such kind can be embedded to a continuous path (see Section \ref{timeSeriesEs}) and the corresponding signature can capture the information in the data stream in a structural way with finite partial descriptions and the leading terms are not particularly sensitive to the sampling rates either (see Section 2.3 in \cite{Greg}). It is important to note that the dimensionality of the truncated signature of a $d$-dimensional path of order $n$ is $\frac{d^{n+1} - 1}{d - 1}$, and it does not depend on the frequency of time sampling. Using the signature features can avoid the curse of dimensionality resulted from the high sampling frequency. For the discussion on sub-sampling continuous data streams using rough paths, we refer readers to \cite{flint2013discretely}.

\subsection{An illustrative example: a diffusion process}
There is a simple but illustrative example - learning $Y_{T}$  (a solution to the SDE driven by multi-dimensional Brownian motion with polynomial vector fields)  as a function of the signature of the driving Brownian path. More specifically, suppose $Y_{t}$ satisfies the following SDE:
\begin{eqnarray*}
dY_{t} = a(1-Y_{t})dX_{t}^{(1)} + bY_{t}^{2}dX_{t}^{(2)}, Y_{0} = 0.
\end{eqnarray*}
where $X_{t} = (X_{t}^{(1)}, X_{t}^{(2)}) = (t, W_{t})$, and the integral is in the Stratonovich sense, and $T = 0.25$, and $(a, b)$ is chosen to $(1, 2)$. This example is borrowed from \cite{papavasiliou2011parameter}, which studied the estimation problem for the parameters $a$ and $b$ of this parametric SDE via the expected signature matching method. However in our paper,  we are interested in non-parametric estimation for the functional relationship $f$ between $X_{[0, T]}$ and $Y_{T}$, where $f$ maps $X_{[0, T]}$ to the corresponding pathwise solution $Y_{T}$ to SDE driven by $X_{[0, T]}$, in formula 
\begin{eqnarray*}
 f(X_{[0, T]}) = Y_{T} .
\end{eqnarray*}
It is noted that  $Y_{T} =  \mathbb{E}[Y_{T} \vert X_{[0, T]}] $, which is the first term of the conditional expected signature of $Y$ given the information $X_{[0, T]}$.

%%%%%%%%%%%%%%%%%%%%%%%%%%%%%%%
We generate 1600 realizations of discretely sampled input path $\hat{X}_{K}:= \left( X_{\frac{iT}{K}} \right)_{i = 0}^{K}$ and the corresponding approximate solution using Milstein method with the number of discretization
steps $K$.  Half of samples are used for training and the rest of samples are used for backtesting. We apply our method to estimate $f$ and benchmark with linear regression on the increment features of $\hat{X}_{K}$, i.e. $\left( X_{\frac{(i+1)T}{K}} - X_{\frac{iT}{K}} \right)_{i = 0}^{K-1}$. 

For $K = 250$, we plot the estimated output for the testing set obtained by linear regression w/t regularization on increment feature set in Figure \ref{LinearRegressionIncrements} and plot the estimated output obtained by linear regression on the truncated signature of different order in Figure \ref{T025Deg2}, Figure \ref{T025Deg4} and Figure \ref{T025Deg6}. The corresponding summary of $R^{2}$ is given in Table \ref{R2_degree_sig}. It shows that increasing the degree of truncated signatures improves the fitting performance. We obtain better forecast results by using the step-4 signature feature of dimension $31$ than the increment feature of dimension 250. The truncated signature of order $6$ gives us almost perfect prediction for $Y_{T}$ in terms of $R^{2}$, and it is a more efficient summary of $X_{[0, T]}$ than the increment feature, which leads to significant dimension reduction.

It is noted that the dimension of the truncated signature of a two-dimensional path $X$ of order $n$ is $2^{n+1} - 1$, independent with the number of sampling time points $K$, while the increment feature of that have the dimensionality $K$.  Figure \ref{RegNoOfTimesteps}  shows that when $K$ is getting larger and larger, the use of increment feature $\hat{X}_{K}$ results in more severe overfitting issue; combing the increment features with the regularization method, like Lasso and cross validation, can help to avoid this problem, but it still underperforms our approach in terms of $R^{2}$ in Table \ref{R2_K} in this SDE example.
\begin{table}\caption{$R^{2}$ summary}\label{R2_degree_sig}
\center
\begin{tabular}{|c|c|c|c|}
\hline
\multicolumn{4}{|c|}{Signature Feature sets} \\
  \hline  \cline{1-2} \cline{3-4} 
  Deg of signature & 2 & 4  &6 \\
  \hline
 R2(backtesting set) &0.9562 & 0.9997 & $\mathbf{1.000}$   \\
  \hline
  \hline
  \multicolumn{4}{|c|}{Increment Feature sets} \\
  \hline 
  &OLS &  \multicolumn{2}{c|}{Lasso + CV}\\
  \hline
     R2(backtesting set) & 0.9673& \multicolumn{2}{c|}{0.9717}\\
     \hline
\end{tabular}
\end{table}

\begin{figure}[ht]
\centering
\subfigure[Linear regression on increments of Brownian paths.]{
 \includegraphics[width= 2.2 in]{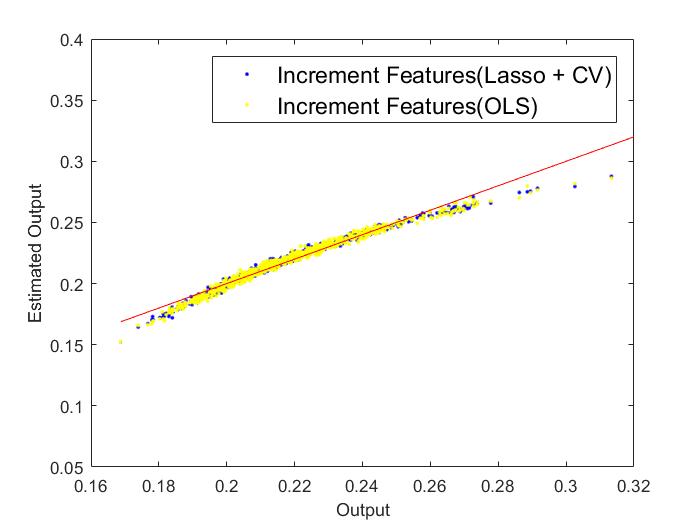}
    \label{LinearRegressionIncrements}
}
\subfigure[Regression on the truncated signature of order $2$.]{
 \includegraphics[width= 2.2 in]{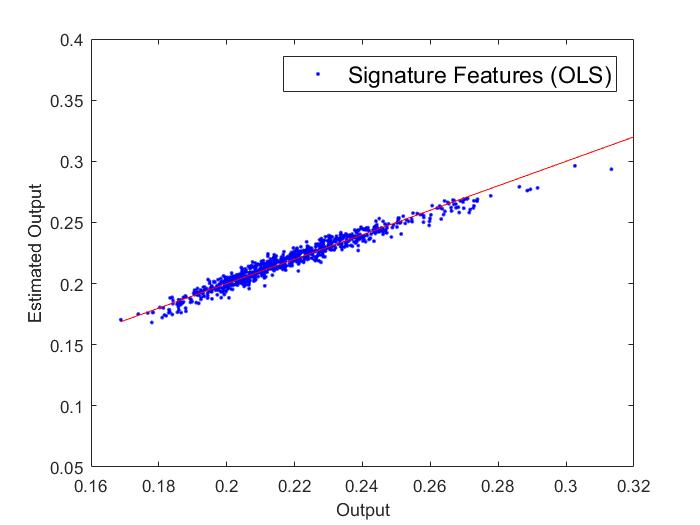}
  %  \rule{4cm}{3cm}
    \label{T025Deg2}
}
\subfigure[Linear regression on the truncated signature of order $4$.]{
 \includegraphics[width= 2.2 in]{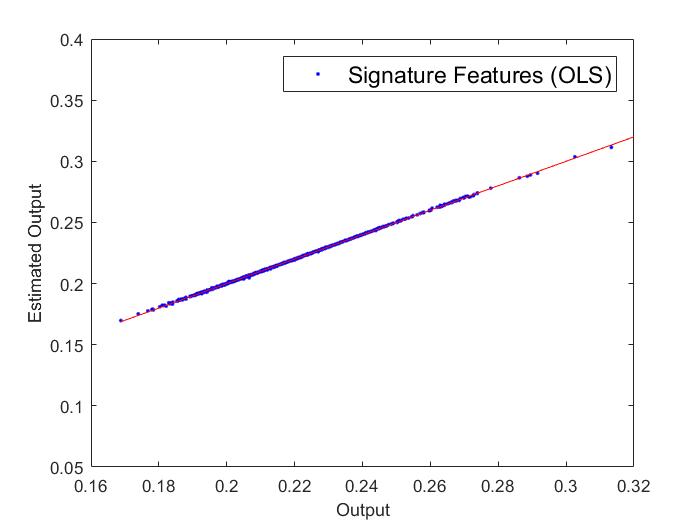}
    \label{T025Deg4}
}
\subfigure[Linear regression on the truncated signature of order $6$.]{
 \includegraphics[width= 2.2 in]{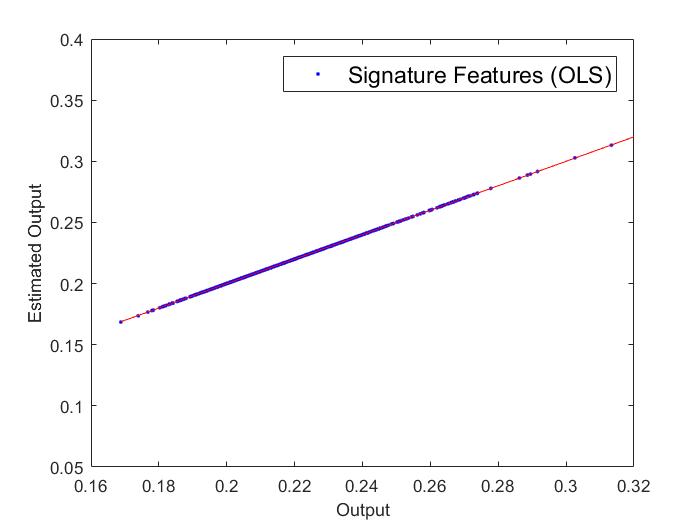}
    \label{T025Deg6}
}
\caption[Optional caption for list of figures]{$K = 250$. In Subfigure \subref{LinearRegressionIncrements}, we plot the estimated output against the actual output via linear regression w/t regularization represented by yellow/blue dots respectively. In Subfigures \subref{T025Deg2}, \subref{T025Deg4} and \subref{T025Deg6} we plot the fitting results obtained via linear regression on the truncated signatures of order $2, 4$ and $6$ respectively.}
\label{FigureRegDegOfSig}
\end{figure}

\begin{figure}[ht]
\centering
\subfigure[Linear regression on increments of Brownian paths. $K = 750$.]{
 \includegraphics[width= 2.2 in]{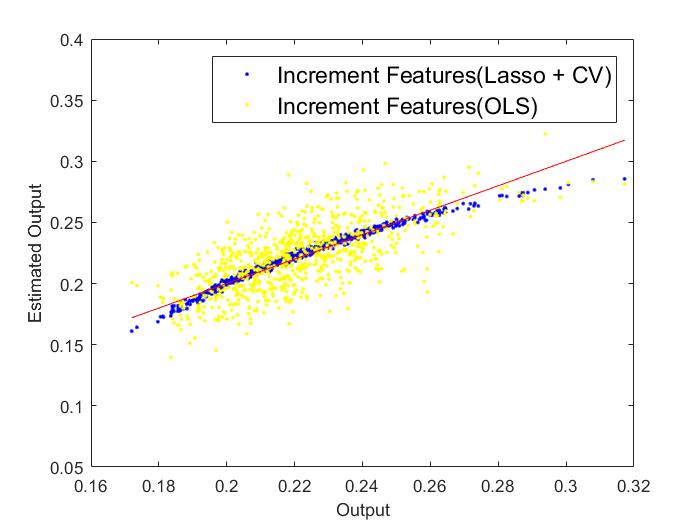}
    \label{timestep750increments}
}
\subfigure[Linear regression on the truncated signature of order $4$. $K = 750$.]{
 \includegraphics[width= 2.2 in]{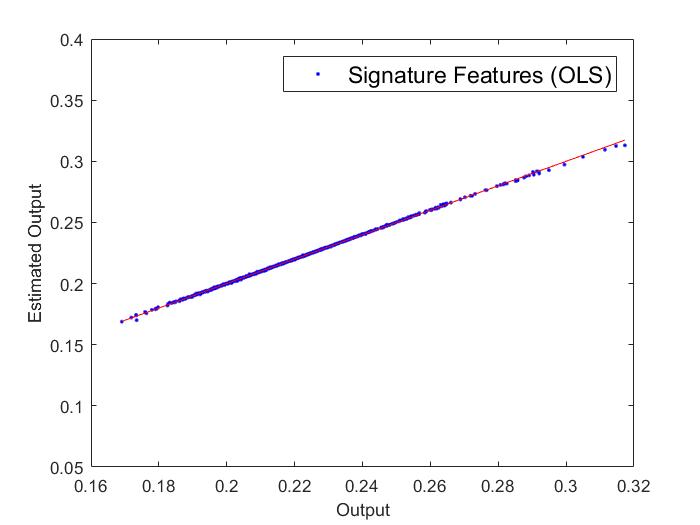}
  %  \rule{4cm}{3cm}
    \label{timestep750sig}
}
\subfigure[Linear regression on increments of Brownian Paths. $K = 1000$.]{
 \includegraphics[width= 2.2 in]{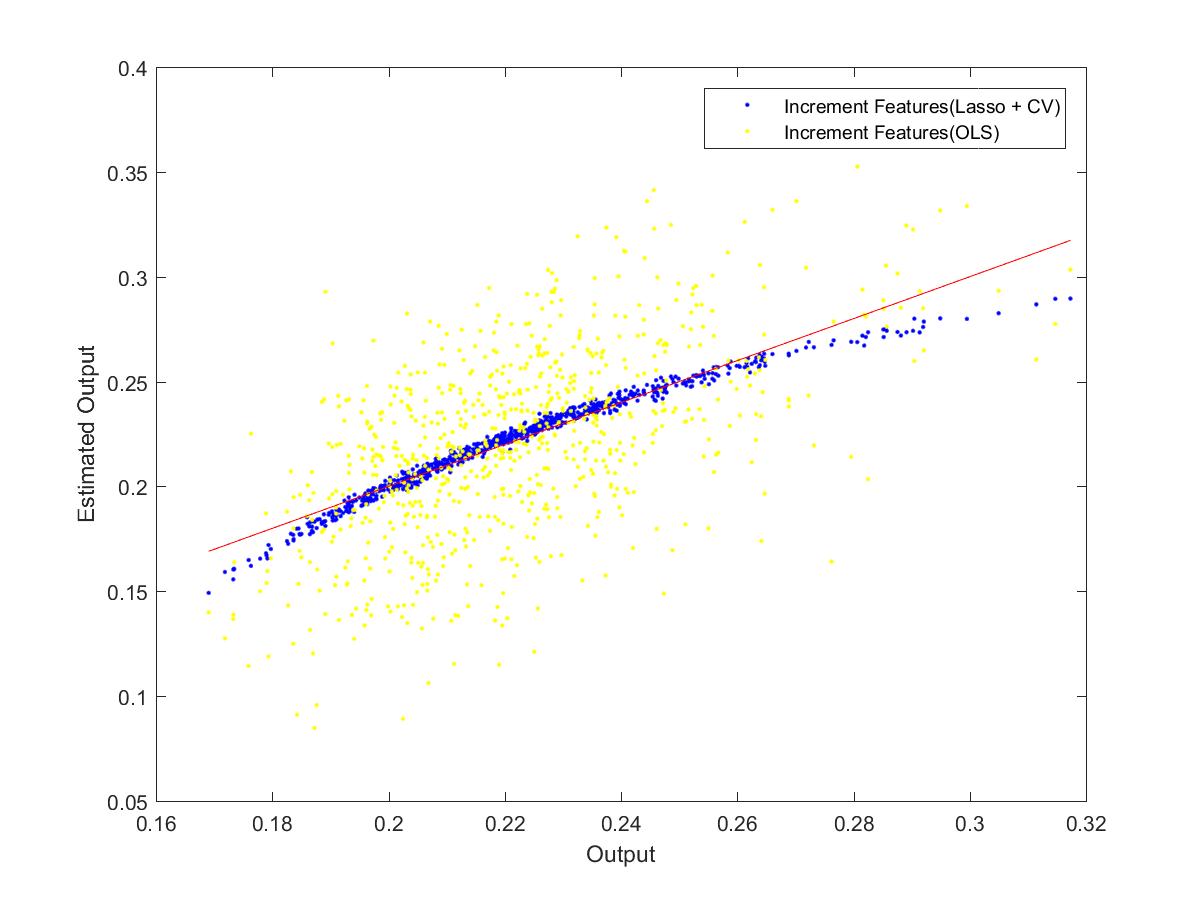}
    \label{timestep1000increments}
}
\subfigure[Regression on the truncated signature of order $4$. $K = 1000$.]{
 \includegraphics[width= 2.2 in]{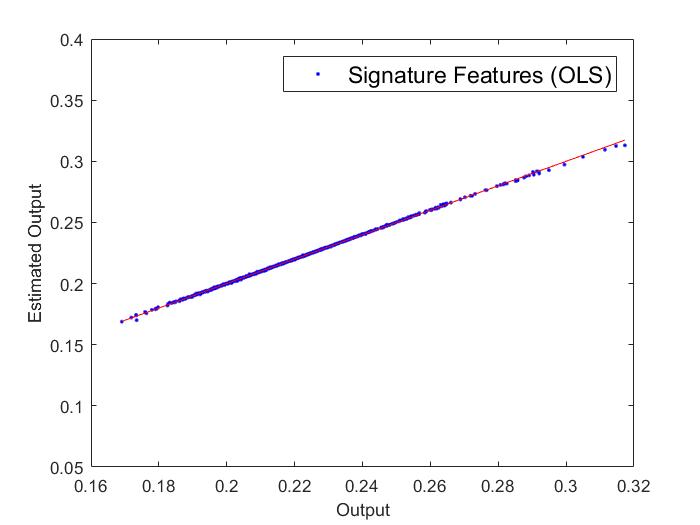}
  %  \rule{4cm}{3cm}
    \label{timestep1000sig}
}
\caption[Optional caption for list of figures]{On the upper panel of figures \subref{timestep750increments} and \subref{timestep750sig}, $K = 750$, while on the lower panel of figures  \subref{timestep1000increments} and \subref{timestep1000sig}, $K = 1000$.  On the left panel, subfigures \subref{timestep750increments} and \subref{timestep1000increments} plot the fitting results for the testing set via linear regression on the increment features, while subfigures  \subref{timestep1000increments} and \subref{timestep750sig} plot the fitting results for the testing results via linear regression on the truncated signature of order $4$. }
\label{RegNoOfTimesteps}
\end{figure}

\begin{table}\caption{$R^{2}$ summary}\label{R2_K}
\center
\begin{tabular}{|c|c|c|c|c|}
\hline
  K& 250& 500& 750& 1000\\
  \hline \hline
  Increment Features (OLS) &0.9562 &  0.9337& 0.1661 &-1.3734 \\ 
  \hline
  Increment Features (Lasso + CV) &0.9717 &  0.9621 & 0.9671 & 0.9633 \\ 
  \hline
  Step-4 Signature Features & $\mathbf{0.9997}$& $ \mathbf{0.9998}$  & $\mathbf{0.9996}$  &$\mathbf{0.9997}$ \\ 
  \hline
\end{tabular}
\end{table}

\section{Time series models}\label{timeSeriesEs}
In this section, we focus on the univariate time series case for simplicity, but our results can be easily generalized for a multi-dimensional time series. Let $\{(t_{i}, r_{i})\}_{i=0}^{N}$ denote a univariate time series, where $t_{i}$ represents the $i^{th}$ time stamp and $r_{i}$ represents the $i^{th}$ data point at time $t_{i}$. We explain how to embed a time series segment to the signature in Subsection \ref{SubsectionSigTimeSeries}. 
\subsection{The signature of a time series}\label{SubsectionSigTimeSeries}
 For any $0 \leq m <n \leq N$, and $m,n \in \mathbb{N}$, in general the signature of $\{(t_{i}, r_{i})\}_{i=m}^{n}$ is computed via the following two steps:
\begin{enumerate}
\item Embed a time series $\{(t_{i}, r_{i})\}_{i = m}^{n}$ into a continuous path;
\item Compute the signature of this transformed continuous path $R$.
\end{enumerate}
There are several choices of embedding, including
\begin{enumerate}
\item[(a)] the piecewise linear interpolation;
\item[(b)] the lead-lag transformation;
\item[(c)] the time-joined path.
 \end{enumerate}
\begin{remark}
In terms of method $(a)$, it might lose the essential information if $E = \mathbb{R}$, since the signature of this transformed path is only the increment at the terminal time, ignoring all the other information. Roughly speaking $(b)$ is the lifted path composed of the original path and the delayed path, whose definition can be found in \cite{Greg}. An advantage of the lead-lag transformation is that one can read the volatility of the path directly from the second term of the signature, which is very important in finance. But in this paper we aim to bridge our model with the classical time series model, and in this setting the time-joined path is a more suitable candidate, which we will explain in the following section. \end{remark}
\begin{remark}
In cases of the higher dimensional time series or high frequency time series, if one wants to summarize the information about time series ignoring time re-parameterizations, one should consider $(a)$ and $(b)$ instead of $(c)$, since our representation of the history through ''the signature'' allows high frequency information to be summarized and incorporated using only a few parameters and so preventing over fitting that is frequently present when trying to use highly sampled data to predict events on longer time intervals.
\end{remark}
\subsubsection{The transformation by retaining the time component}
\begin{definition}[Time-joined transformation]
Let $\{(t_{i}, r_{i})\}_{i = m}^{n}$ be a univariate time series. Let $R: [2m, 2n+1] \rightarrow \mathbb{R}^{+} \times \mathbb{R}$ be a $2$-dimensional time-joining path of $\{(t_{i}, r_{i})\}_{i = m}^{n}$, which is defined as follows:
\begin{eqnarray*}
R(s) = \begin{cases}   t_{m}e_{1}+r_{m}(s-2m)e_{2}, & \text{if $s \in [2m, 2m+1)$;}\\
[ t_{i}+(t_{i+1}- t_{i})(s-2i-1)]e_{1}+r_{i}e_{2}, & \text{if $s \in [2i+1, 2i+2)$;}\\
  t_{i+1}e_{1}+[r_{i}+(r_{i+1}-r_{i})(s-2i-2)]e_{2}, & \text{if $s \in [2i+2, 2i+3)$;}
\end{cases}
\end{eqnarray*}
where $i = m, m+1, ..., n-1$  and $\{e_{i}\}_{i = 1, 2}$ is an orthonormal basis of $ \mathbb{R}^{2}$.
\end{definition}
\begin{remark}
The continuous function $R$ is simply defined as keeping $r_{i}$ value at the time interval $[t_{i}, t_{i+1})$. When the new data $r_{i+1}$ arrives at time $t_{i+1}$, there is an instantaneous jump from $r_{i}$ to $ r_{i+1}$. We add one more point $0$ at the time $t_{m}$ to the time series $\{r_{i}\}_{i=m}^{n}$ to make it a new time series, such that the signature of  $\{r_{i}\}_{i=m}^{n}$ can uniquely determine $\{r_{i}\}_{i=m}^{n}$(Lemma \ref{Lemma_sig_determine_path}).
\end{remark}
The signature of $\{(t_{i}, r_{i})\}_{i = m}^{n}$ is defined to be the signature of the time-joined transformation of $\{(t_{i}, r_{i})\}_{i = m}^{n}$.
\begin{definition}\label{timeSeriesDef}
Let $\{(t_{i}, r_{i})\}_{i = m}^{n}$ be a time series and embedded into time-joined path $R$. The signature of the time series $\{(t_{i}, r_{t_{i}})\}_{i = m}^{n}$ is defined by the signature of the path $\{R(s)\}_{s \in [2m, 2n+1]}$ and we denote it by $S(\{(t_{i}, r_{i})\}_{i = m}^{n})$, where $ 1 \leq m < n \leq N$ and $ m,n \in \mathbb{N}$.
\end{definition}
%%%%%%%%%%%%%%%
Let us consider a simple example, and illustrate how to obtain a continuous function $R$ from the time series via the figure.
\begin{example}
Suppose that a univariate time series is given as
\begin{eqnarray*}
\{(2,2), (3,5), (4, 3), (5, 4) , (6, 6), (7, 3) ,(8, 2)\},
\end{eqnarray*}
and plotted below as ${\color{green} *}$, then the function $R$ is shown in the blue curve in Figure \ref{timeSeriesPath}.
\end{example}
\begin{figure}[timeSeriesPath]
  \centering
    \includegraphics[width=3.2in]{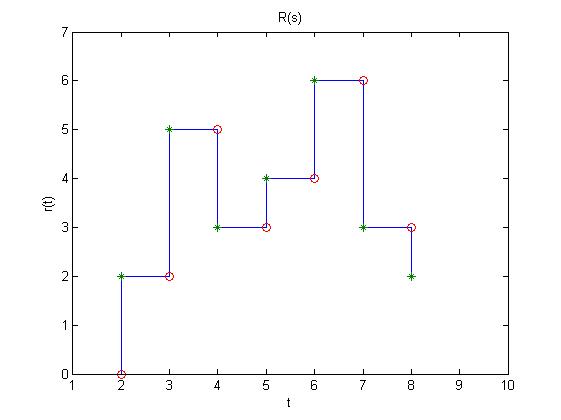}
    \caption{Embedding the time series into the continuous function $R$.}\label{timeSeriesPath}
\end{figure}
\begin{lemma}
The signature of a time series $\{(t_{i}, r_{i})\}_{i = m}^{n}$ is
\begin{eqnarray*}
 \exp(r_{1} e_{2}) \overset{n-1}{\underset {i = m} {\bigotimes}} \left( \exp((t_{i+1}-t_{i})e_{1}) \otimes \exp((r_{i+1}-r_{i})e_{2}) \right),
\end{eqnarray*}
where $ 1 \leq m < n \leq N$ and $ m,n \in \mathbb{N}$.
\end{lemma}
\begin{proof}
Since the transformed path $R$ is piecewise linear, by induction using Chen's identity we will have this lemma result immediately.
\end{proof}

\subsubsection{The properties of the signature of a time series}
Now we discuss the main properties of the signature of a time series $\{(t_{i}, r_{i})\}_{i = m}^{N}$.
\begin{lemma}\label{Lemma_sig_determine_path}
Suppose that $0 < m < n \leq N, m, n \in \mathbb{N}$. The signature of a time series $\{(t_{i}, r_{i})\}_{i = m}^{n}$  uniquely determines the time series $\{(t_{i}, r_{i})\}_{i = m}^{n}$.
\end{lemma}
\begin{proof}
Since the function $\pi^{1}\circ R$ is a non-decreasing function, the function $R$ can't be tree-like. Thus there is one-to-one correspondence between $\{t_{i}, r_{i})\}_{i = m}^{n}$ and the signature of  $\{(t_{i}, r_{i})\}_{m}^{n}$ since the ending point of the function $R(.)$ is given that $R(2m) = t_{m}e_{1}$.
\end{proof}
The following lemma states that $\{r_{i}\}_{i=1}^{n}$ can be represented as a linear functional on its signature of $\{(t_{i}, r_{i})\}_{i = 1}^{n}$.
\begin{lemma}\label{LemmaLinearSig}
Let $\mathbf{X}$ denote the signature of a time series $\{(t_{i}, r_{i})\}_{i = 1}^{n}$. Assume that $\{t_{i}\}_{i=1}^{n}$ are known. Then
\begin{eqnarray*}
\Delta R= T^{-1}A,
\end{eqnarray*}
where
\begin{eqnarray*}
&&A :=  \left( \begin{array}{c}
0!\pi^{2}(\mathbf{X}) \\
1!\pi^{12}(\mathbf{X}) \\
\vdots\\
(n-1)!\pi^{1\dots 12}(\mathbf{X}) \end{array} \right), T: =  \left( \begin{array}{ccccc}
1&1  & \dots& 1\\
t_{1}&t_{2}  & \dots &t_{n}\\
\vdots&\vdots & \ddots & \vdots\\
t_{1}^{n-1}& t_{2}^{n-1} &\dots & t_{n}^{n-1}\end{array} \right), \\
&&\Delta R:=  \left( \begin{array}{c}
r_{1}\\
r_{2} - r_{1}  \\
\vdots\\
r_{n} - r_{n-1}
\end{array} \right).
\end{eqnarray*}
\end{lemma}
\begin{proof}
Assume that $\{t_{i}\}_{i = 1}^{n}$ are known. Consider the coordinate signature indexed by $(\underset{k-1}{\underbrace{1, 1, \dots, 1}}, 2)$, where there are $k-1$ copies of 1. We have that
\begin{eqnarray*}
\pi^{(1, \dots, 1, 2)}(\mathbf{X}) = \int_{0}^{2n-1} \frac{1}{k!} \{R_{s}^{(1)}\} ^{k}dR_{s}^{(2)}.
\end{eqnarray*}
Let us recall the definition of $R_{s}$:
\begin{eqnarray*}
R(s) = \begin{cases}   t_{m}e_{1}+r_{m}(s-2m)e_{2}, & \text{if $s \in [2m, 2m+1)$;}\\
[ t_{i}+(t_{i+1}- t_{i})(s-2i-1)]e_{1}+r_{i}e_{2}, & \text{if $s \in [2i+1, 2i+2)$;}\\
  t_{i+1}e_{1}+[r_{i}+(r_{i+1}-r_{i})(s-2i-2)]e_{2}, & \text{if $s \in [2i+2, 2i+3)$;}
\end{cases}
\end{eqnarray*}
where $i = m, m+1, ..., n-1$  and $\{e_{i}\}_{i = 1,2}$ is an orthonormal basis of $\mathbb{R}^{2}$.\\
So $\pi^{(1, \dots, 1, 2)}(\mathbf{X})$ can be simplified to
\begin{eqnarray*}
\pi^{(1, \dots, 1, 2)}(\mathbf{X}) &=&  \int_{0}^{1}\frac{1}{k!}t_{1}^{k}d(r_{t_{1}}s)+ \sum_{i=1}^{n-1} \int_{2i}^{2i+1}\frac{1}{k!}t_{i}^{k}d((r_{t_{i+1}}-r_{t_{i}})s) \\
&=&  \frac{1}{k!}t_{1}^{k}r_{t_{1}}+ \sum_{i=1}^{n-1}\frac{1}{k!}t_{i}^{k}(r_{t_{i+1}}-r_{t_{i}}).\\
\end{eqnarray*}
For $n \geq 2$,
it can be rewritten in the matrix form:
\begin{eqnarray*}
A = T \Delta R,
\end{eqnarray*}
where
\begin{eqnarray*}
&&A :=  \left( \begin{array}{c}
0!\pi^{(2)}(\mathbf{X}) \\
1!\pi^{(1, 2)}(\mathbf{X}) \\
\vdots\\
(n-1)!\pi^{(1, \dots, 1, 2)}(\mathbf{X}) \end{array} \right), T: =  \left( \begin{array}{ccccc}
1&1  & \dots& 1\\
t_{1}&t_{2}  & \dots &t_{n}\\
\vdots&\vdots & \ddots & \vdots\\
t_{1}^{n-1}& t_{2}^{n-1} &\dots & t_{n}^{n-1}\end{array} \right), \\
&&\Delta R:=  \left( \begin{array}{c}
r_{1}\\
r_{2} - r_{1}  \\
\vdots\\
r_{n} - r_{n-1}
\end{array} \right).
\end{eqnarray*}
Since $T$ is a square Vandermonde matrix, and all $t_{i}$ are different from each other for $i=0, 1, \dots n-1$, it is invertible. So we have
\begin{eqnarray*}
\Delta R = T^{-1}A.
\end{eqnarray*}
\end{proof}
\subsection{The ES model for time series}
Let $\{r_{i}\}_{i=1}^{N}$ denote a univariate time series. Fix two positive integers $p$ and $q$. For the fixed $k \in \mathbb{N}$, denote the information set available at time $t_{k}$, i.e. the past returns before $t_{k}$ by $\mathcal{F}_{k}$. In this context, $\mathbf{Y}_{k}$ is the signature of the future return series $S(\{t_{i}, r_{i}\}_{i =k+1}^{k+q} )$, and $\mathbf{X}_{k}$ is the signature of the past return series $S(\{t_{i}, r_{i}\}_{i =k-p}^{k} )$.  Let us introduce the ES model as follows.
\begin{definition}[$ES(p,q,n,m)$]
Suppose that a univariate time series $\{r_{i}\}_{i=1}^{N}$ is stationary. We say that $\{r_{i}\}_{i=1}^{N}$ satisfies the assumptions of the $ES$ model with parameters $p, q, n$ and $m$, denoted by $ES(p,q,n,m)$ if there exists a linear function $f: T^{n}(\mathbb{R}^{2}) \rightarrow T^{m}(\mathbb{R}^{2})$ such that
\begin{eqnarray*}
\rho_{m}(S(\{r_{t+i}\}_{i=1}^{q})) = f(\rho_{m}(S(\{r_{t-i}\}_{i=0}^{p})))+a_{t},
\end{eqnarray*}
where $N$ is a positive integer such that $N \geq p+q$, and the residual terms $a_{t}$ satisfy
\begin{eqnarray*}
\mathbb{E}[a_{t} \vert \mathcal{F}_{t}]=0.
\end{eqnarray*}
\end{definition}
Let $\mu_{k}$ denote the expectation of $S(\{t_{i}, r_{i}\}_{i =k+1}^{k+q} )$ conditional on the information up to the time $t_{k}$, i.e.
\begin{eqnarray}
\mu_{k} &=& \mathbb{E}[S(\{t_{i}, r_{i}\}_{i =k+1}^{k+q} ) \vert \mathcal{F}_{k} ],\label{MeanEquation}
\end{eqnarray}
and $\mu_{k}$ is a function on $\mathbf{X}_{k}$, i.e. there exists a function $f: T((\mathbb{R}^{2})) \rightarrow T((\mathbb{R}^{2}))$ such that
\begin{eqnarray*}
\mu_{k} = f(\mathbf{X}_{k}).
\end{eqnarray*}
Observe that $\mu_{k}$ and $a_{k}$ take values in $T((E))$. The conditional covariance of the signature of the future return series $S(\{t_{i}, r_{i}\}_{i =k+1}^{k+q} )$ given $\mathcal{F}_{k}$ is defined as the function $\Sigma_{k}^{2}: A^{*} \times A^{*} \rightarrow \mathbb{R}$:
\begin{eqnarray}
\Sigma_{k}^{2}(I, J) =  \mathrm{Cov}(\pi^{I}(\mathbf{Y_{k}}) , \pi^{J}(\mathbf{Y_{k}}) \vert \mathcal{F}_{k}),\label{VolatilityEquation}
\end{eqnarray}
where $I, J \in A^{*}$.

The model for $\mu_{k}$ in Equation (\ref{MeanEquation}) is referred to the mean equation for $\mathbf{Y}_{k}$ and the model for $\Sigma_{k}^{2}$ in Equation (\ref{VolatilityEquation}) is the covariance equation for $\mathbf{Y}_{k}$. Correspondingly we use $\mu_{k}^{(m)}$ and $(\Sigma_{k}^{2})^{(m)}$ to denote the expectation and the covariance of the truncated signature $\rho_{m}(\mathbf{Y_{k}})$ conditional on the information up to the time $t_{k}$ respectively. Lemma \ref{MeanDeterminesVar} shows that $\mu_{k}$ determines $\Sigma_{k}^{2}$ due to the shuffle product property of the signature. In the ES model, by assigning the model for the mean equation for the truncated signature of the future returns of the order $2n$, it automatically determines the conditional variance structure of the truncated signature of that of order $n$. 

The fundamental assumption of the ES model is the stationarity of the time series $\{r_{i}\}$, which is standard in the time series analysis. A time series $\{r_{k}\}$ is said to be strictly stationary if the joint distribution of $(r_{t_{1}}, ..., r_{t_{k}})$ is identical to that of  $(r_{t_{1}+\tau}, ..., r_{t_{k}+\tau})$ for any $\tau \in \mathbb{R}^{+}$, where $k$ is an arbitrary positive integer and $(t_{1}, \dots t_{k})$ is a collection of $k$ positive integers (see \cite{AnalysisofFinancialTimeSeries}). It implies that the distribution of the signature of  $(r_{t_{1}}, ..., r_{t_{k}})$ is invariant under the time shift as well.

The $ES(p, q, n, m)$ model assumes that the distribution of $ r_{k+1}, \dots, r_{k+q}$ on condition of the current information $ \mathcal{F}_{k}$ only depends on the truncated signature of the $p$-lagged data points $r_{k-p}, \dots, r_{k-1}, r_{k}$ of the order $n$, which are rich enough to approximate any smooth mean function on $p$-lagged data. The conditional expectation of signature of  $r_{k+1}, \dots, r_{k+q}$ provides a non-parametric way to characterise the distribution of $r_{k+1}, \dots, r_{k+q}$ given the information $\mathcal{F}_{k}$, following the discussion about the expected signature model framework in Section \ref{GeneralExpectedSignature}.

\subsection{The link between the classical time series models and the ES model}
There are various autoregressive-type models in statistics and econometrics, e.g. Autoregressive (AR) model and Autoregressive Conditional Heteroskedasticity (ARCH) models. (The definition of AR and ARCH models can be found in the appendix.) These time series models focused on modeling and estimation of the conditional expectation and variance of the future data $r_{k+1}$ given the information up to time $t_{k}$, which are denoted by $m_{k}$ and $\sigma_{k}^{2}$ respectively, i.e.
\begin{eqnarray*}
m_{k} &:=& \mathbb{E}[r_{k+1}\vert \mathcal{F}_{k} ], \\
\sigma_{k}^{2} &:=& \mathrm{Var}[r_{k+1} \vert \mathcal{F}_{k} ].
\end{eqnarray*}
$m_{k}$ and $\sigma^{2}_{k}$ are referred to the mean equation for $r_{k}$ and the variance equation for $r_{t}$, which can be obtained by $\mu_{k}$ - the expected signature of $\{(t_{k+i}, r_{k+i})\}_{i = 0}^{q}$ on condition of $\mathcal{F}_{k}$, where $q = 1$. That's because
\begin{eqnarray*}
m_{k} &=& \pi^{(2)}(\mu_{k}), \\
\sigma_{k}^{2} &=& 2\pi^{(2, 2)}(\mu_{k})- (\pi^{(2)}(\mu_{k}))^{2}.
\end{eqnarray*}

As linear forms on the signature of $p$-lagged values of $r_{t}$ are dense in the space of smooth functions on  $p$-lagged values of $r_{t}$, it incorporates the classical time series, like AR, ARCH and so on. Next we use the ARCH model as an example to show that it is a special case of the ES model. To do so, we need the following auxiliary lemma about the moments of the return $r_{k}$ conditional on the information up to time $k-1$.
\begin{lemma}\label{Lemma_moment_ARCH}
Suppose that a time series $\{r_{k}\}$ satisfies $ARCH(q)$ model given in Definition $\ref{Def_ARCH}$. $\mu_{k}$ is the mean equation in the form as follows:
\begin{eqnarray*}
\mu_{k} = \beta_{0} + \sum_{i = 1}^{Q} \beta_{i} r_{k-i},
\end{eqnarray*}
where $Q$ is a positive integer, $\{\beta_{i}\}_{i=0}^{Q}$ are all constants.
Suppose $z_{k}$ is a white noise satisfying the condition that $z_{k}$ has the moments up to degree $n$ and $\mathbb{E}[z_{k}^{j}] = 0$ if $j$ is an odd integer. Then for every positive integer $n \geq 1$, $\mathbb{E}[r_{k}^{n} \vert \mathcal{F}_{k-1}] $ is a polynomial of lagged $(q+Q)$ values of $r_{k}$.
\end{lemma}
\begin{proof}
For every positive integer $n \geq 0$ such that $z_{k}$ has finite moments up to degree $n$, it holds that
\begin{eqnarray}
\mathbb{E}[r_{k}^{n} \vert \mathcal{F}_{k-1}] &=& \mathbb{E}[(\mu_{k} + \sigma_{k}z_{k})^{n} \vert \mathcal{F}_{k-1}]\label{eq1}\\
&=& \sum_{j=0}^{n} \mathbb{E}[C_{n}^{j} \mu_{k}^{n-j}(\sigma_{k}z_{k})^{j}\vert \mathcal{F}_{k-1}]\label{eq2}\\
&=& \sum_{j=0}^{n} C_{n}^{j} \mu_{k}^{n-j} \sigma_{k}^{j}\mathbb{E}[z_{k}^{j}];\label{eq3}
\end{eqnarray}
where $C_{n}^{j}$ is the coefficient of the monomial $x^{k}$ in the expansion of $(1+x)^{n}$, i.e. $C_{n}^{j} = \frac{n!}{(n-j)!j!}$ for any integers $j$ and $n$ such that $0 \leq j \leq n$.\\
(\ref{eq1}) comes from the definition of residual terms $\varepsilon_{k} = r_{k}-\mu_{k} = \sigma_{k}z_{k}$, while (\ref{eq2}) and (\ref{eq3}) follows using the linearity of the expectation and the property of white noise $z_{k}$.
By assumption of $z_{k}$, $\mathbb{E}[z_{k}^{j}]$ is zero for all odd $j\leq n$, and thus
\begin{eqnarray*}
\mathbb{E}[r_{k}^{n} \vert \mathcal{F}_{k-1}]= \sum_{j=0 \text{ and } j \text{ is even}}^{n} C_{n}^{j} \mu_{k}^{n-j} \sigma_{k}^{j}\mathbb{E}[z_{k}^{j}].
\end{eqnarray*}
Using the definition of the error term $\varepsilon_{k-i} = r_{k-i} - \mu_{k-i}$ and the variance equation
\begin{eqnarray*}
\sigma_{k}^{2} = \alpha_{0} + \sum_{i = 1}^{q}\alpha_{i} \varepsilon_{k-i}^{2},
\end{eqnarray*}
for any even $j \leq n$, $\sigma_{k}^{j}$ is obviously a polynomial of the $(q+Q)$-lagged values of $r_{k}$, and so is $\mathbb{E}[r_{k}^{n} \vert \mathcal{F}_{k-1}]$, since
\begin{eqnarray*}
\sigma_{k}^{j} = \left( \alpha_{0} + \sum_{i = 1}^{q}\alpha_{i} (r_{k-i} - \mu_{k-i})^{2}\right)^{\frac{j}{2}},
\end{eqnarray*}
where
\begin{eqnarray*}
\mu_{k-i} = \beta_{0} + \sum_{j = 1}^{Q} \beta_{j} r_{k-i-j}, \forall i = 1, \cdots, q.
\end{eqnarray*}
\end{proof}
\begin{theorem}\label{Thm_ARCH_ES}
Suppose that a time series $\{r_{k}\}$ satisfies the assumptions of the $ARCH(q)$ model given in Definition $\ref{Def_ARCH}$ and its mean equation $\mu_{k}$ is in the following form:
\begin{eqnarray*}
\mu_{k} = \beta_{0} + \sum_{i = 1}^{Q} \beta_{i} r_{k-i}.
\end{eqnarray*}
Then there exists a sufficiently large integer $N$ such that a time series $\{r_{k}\}$ satisfies the assumptions of $ES(q+Q, 1, N, 2)$.
\end{theorem}
\begin{proof}
It is equivalent to check whether $\mathbb{E}[r_{k}^{n} \vert \mathcal{F}_{k-1}]$ can be expressed as a linear functional on the signature of the $(q+Q)$-lagged values of $r_{k}$ for $n = 1, 2$. By Lemma \ref{Lemma_moment_ARCH}, for $n=1, 2$, $\mathbb{E}[r_{k}^{n} \vert \mathcal{F}_{k-1}]$ have an explicit representation as follows:
\begin{eqnarray*}
\mathbb{E}[r_{k} \vert \mathcal{F}_{k-1}] &=& \mu_{k} = \beta_{0} + \sum_{i = 1}^{Q} \beta_{i} r_{k-i},\\
\mathbb{E}[r_{k}^{2} \vert \mathcal{F}_{k-1}] &=& \alpha_{0} + \sum_{i = 1}^{q}\alpha_{i} (r_{k-i} - \mu_{k-i})^{2}.
\end{eqnarray*}
Therefore for $n = 1, 2$, $\mathbb{E}[r_{k}^{n} \vert \mathcal{F}_{k-1}]$ is a polynomial of the $(q+Q)$-lagged values of $r_{t}$, which implies that there exists a linear functional $f_{n}$, such that $\mathbb{E}[r_{k}^{n} \vert \mathcal{F}_{k-1}] = f_{n}(S(\{t_{k -i}, r_{k-i}\}_{i = 1}^{q+Q}))$, due to Theorem \ref{shuffle_theorem} and Lemma \ref{Lemma_sig_determine_path}.
\end{proof}
\begin{remark}
Similarly it is easy to show that AR and GARCH models are special cases of the ES model as well. $ES(p, 1, n, 2)$ model can be simply considered as the classical time series models using the coordinate signature of the past return series as explanatory variables.
\end{remark}
\section{Examples and numerical computation}\label{NumericalExamples}
\subsection{The data and experimental setup}
In this section, we consider the time series $\{r_{t}\}_{t}$ generated by the following equation
\begin{eqnarray*}
r_{t+1} = m_{t} + \sigma \varepsilon_{t},
\end{eqnarray*}
where $\varepsilon_{t}$ are standard white noise, the volatility is a constant $\sigma$, and the mean equation $m_{t}$ has the following three types:
\begin{enumerate}
\item (AR model)
\begin{eqnarray*}
m_{t}= \Phi_{0} + \sum_{i =1}^{p}\Phi_{i}r_{t-i+1},
\end{eqnarray*}
where $\{\Phi_{i}\}_{i =0}^{p}$ are constant;
\item  (Poly AR model)
\begin{eqnarray*}
m_{t} = f(r_{t}, r_{t-1}, \dots, r_{t-p}) ,
\end{eqnarray*}
where $f$ is a polynomial of degree no less than $2$ in $r_{t}, r_{t-1}, \dots, r_{t-p}$;
\item (Mixture of Poly ARs model)
\begin{eqnarray*}
m_{t} =  \left\{
            \begin{array}{ll}
              f_{1}(r_{t}, r_{t-1}, \dots, r_{t-p}), & \hbox{if $r_{t}> c$;} \\
              f_{2}(r_{t}, r_{t-1}, \dots, r_{t-p}) , & \hbox{if $r_{t}\leq c$.}
            \end{array}
          \right.
\end{eqnarray*}
where $f_{1}$ an $f_{2}$ are both polynomials and $c$ are constant.
\end{enumerate}
Let $\{r_{t}\}_{t = 0}^{N}$ be a time series generated by one of the above models. We use the following three approaches to calibrate $\{r_{t}\}_{t = 0}^{N}$:
\begin{enumerate}
 \item (AR approach) Apply the linear regression of  $r_{t+1}$ against the $p$-lagged values;
 \item (GP approach\footnote{The matlab code of GP regression we used here is written by Carl Edward Rasmussen and Hannes Nickisch and can be downloaded via $http://www.gaussianprocess.org/gpml/code/matlab/doc/$. We own special thanks to Syed Ali Asad Rizvi, who helped us implementing GP regression.  Originally we used the GPfit R-package to implement GP method, but it took a very long time when the learning size is over 200. } ) Apply the Gaussian Process regression of $r_{t+1}$ against the $p$-lagged values. Here we use the exact inference with the Squared Exponential covariance function. Please refer to (\cite{rasmussen2006gaussian}) for more details of the GP regression;
 \item (ES approach) Apply the linear regression of $r_{t+1}$ against the signature of the the $p$-lagged values.
\end{enumerate}
As a parametric approach, the AR approach is standard in removing the mean equation of the AR, ARCH and GARCH models, while as a Bayesian non-parametric method, the GP approach is popular in the fields of machine learning. Therefore those two approaches  are natural candidates to benchmark the performance of the ES approach.

We use the cross-validation(repeated random sub-sampling validation) to measure the predictive error; more precisely, randomly select $M$ of the observations to hold out for the evaluation set, and use the remaining observations for fitting, and repeat this procedure for $N_{cv}$ times. There are several measure of the goodness fit given as follows:
\begin{enumerate}
\item The $R^{2}$ and the adjusted $R^{2}$;
\item MSE (Here it is noticed that the benchmark of the estimated condition mean given the information up to $t$ is true condition mean $m_{t}$ instead of $r_{t+1}$.)
\begin{eqnarray*}
\frac{1}{N}\sum_{t = 1}^{N} (\hat{r}_{t+1} - m_{t})^{2},
\end{eqnarray*}
where $\hat{r}_{t+1}$ is the estimated mean of $r_{t+1}$ on condition of $\mathcal{F}_{t}$ and $N$ is the size of the testing set.
\item The running time.
\end{enumerate}
All the numerical tests are implemented in Matlab7(Release 2014a) using single threaded code on one hardcore.The signatures are computed using the sigtools Python package, which is based on the the libalgebra library of the CoRoPa project2\footnote{Version 0.3, ref.:$\emph{http://coropa.sourceforge.net/}$}, and integrated to matlab.
%%%%%%%
\subsection{Numerical results}
We generate three time series $\{r^{(i)}_{t}\}$ of length $4000$ where $i = 1, 2, 3$ as follows:
\begin{enumerate}
\item $\{r^{(1)}_{t}\}_{t = 0}^{4000}$ satisfies $AR(3)$ with parameters
\begin{eqnarray*}
[\Phi_{0}, \Phi_{1}, \Phi_{2}, \Phi_{3}] = [0, 0.6, 0.15, -0.1].
\end{eqnarray*}
\item $\{r^{(2)}_{t}\}_{t = 0}^{4000}$ satisfies the Poly AR model with the mean equation as follows:
\begin{eqnarray*}
m_{t} = 0.2 r_{t -2}^{(2)}+ 0.1r_{t}^{(2)}(r_{t-1}^{(2)}-r_{t}^{(2)}).
\end{eqnarray*}
\item $\{r^{(3)}_{t}\}_{t = 0}^{4000}$ satisfies the mixture of two Poly AR models as follows:
\begin{eqnarray*}
m_{t} =  \left\{
            \begin{array}{ll}
              -0.6r_{t-2}^{(3)} -0.15r_{t-1}^{(3)} +0.4r_{t}^{(3)}-0.015\left(r_{t-1}^{(3)}\right)^{2}, & \hbox{if $r_{t}^{(3)}> 0$;} \\
             -0.6r_{t-2}^{(3)} -0.15r_{t-1}^{(3)} +0.8r_{t}^{(3)}-0.015\left(r_{t-1}^{(3)}\right)^{2} , & \hbox{if $r_{t}^{(3)}\leq 0$.}
            \end{array}
          \right.
\end{eqnarray*}
\end{enumerate}
In the simulation of three time series $\{r^{(i)}\}$, we deliberately keep all the parameters the same for those three data sets, except for the choice of the mean equation. We apply the AR approach, the ES approach and the GP approach to calibrate those three data sets. For each time series data, we use the first $80\%$ samples for training, and hold on the rest of samples for backtesting. \\
\\
\textbf{Calibration:} Take the first data set $\{r^{(1)}_{t}\}_{t = 0}^{4000}$ as an example, and corresponding numerical results are outlined here. The AR calibration gives the estimator for $\Phi$ as follows:
\begin{eqnarray*}
\hat{\Phi} := [\hat{\Phi}_{0}, \hat{\Phi}_{1}, \hat{\Phi}_{2}, \hat{\Phi}_{3}] =[0.001464, 0.6264, 0.13848, -0.10571].
\end{eqnarray*}
In our model, we specify $n = 4$ and give all possible indices $I$ and the corresponding estimated coefficients $\{\hat{f}^{I}\}$ as follows:
\begin{eqnarray*}\small
\begin{array}{ccccccc}
\hline
I & ()  & (1) & (2) & (1, 1) &(1,2) & (2,1) \\
\hat{f}^{I} &0&0.0044097&0&0& 0& 0  \\
\hline
I &(2,2)&(1,1,1) &(1,1,2) &(1,2,1) &(1,2,2)& (2,1,1)\\
\hat{f}^{I} & 0& 0 &0.32412&-0.10798&0&-0.055639\\
\hline
I & (2,1,2)& (2,2,1)&(2,2,2)& (1,1,1,1)& (1,1,1,2) &(1,1,2,1)\\
\hat{f}^{I}& 0.011852& 0 & 0 &0 & 0 &0\\
\hline
I & (1,1,2,2) &(1,2,1,1)& (1,2,1,2) &(1,2,2,1)& (1,2,2,2) &(2,1,1,1) \\
\hat{f}^{I}&  -0.019671& 0& 0.025391 &0.028881& 0.015265 & 0\\
\hline
I& (2,1,1,2)&(2,1,2,1)& (2,1,2,2) &(2,2,1,1) &(2,2,1,2)&(2,2,2,1)\\
\hat{f}^{I} & 0&0.0097624& 0.0044422& 0.0020564& - 0.0018335& -0.0026311\\
\hline
I&  (2,2,2,2)&&&& &\\
\hat{f}^{I}& 0.011134&&&& &\\
\hline
\end{array}
\end{eqnarray*}
The GP approach gives the fitted parameters $\ln(\lambda) = 3.5342, \ln(h) = 2.3110$ for the squared exponential covariance function, where $\lambda$ and $h$ are the scale for the input and the output respectively. For the other two data sets, the estimated parameters based on those three methods are not given here, but it can be found in Appendix \ref{AppendixParameter}.
\begin{figure}
  \centering
  \includegraphics[width=4 in]{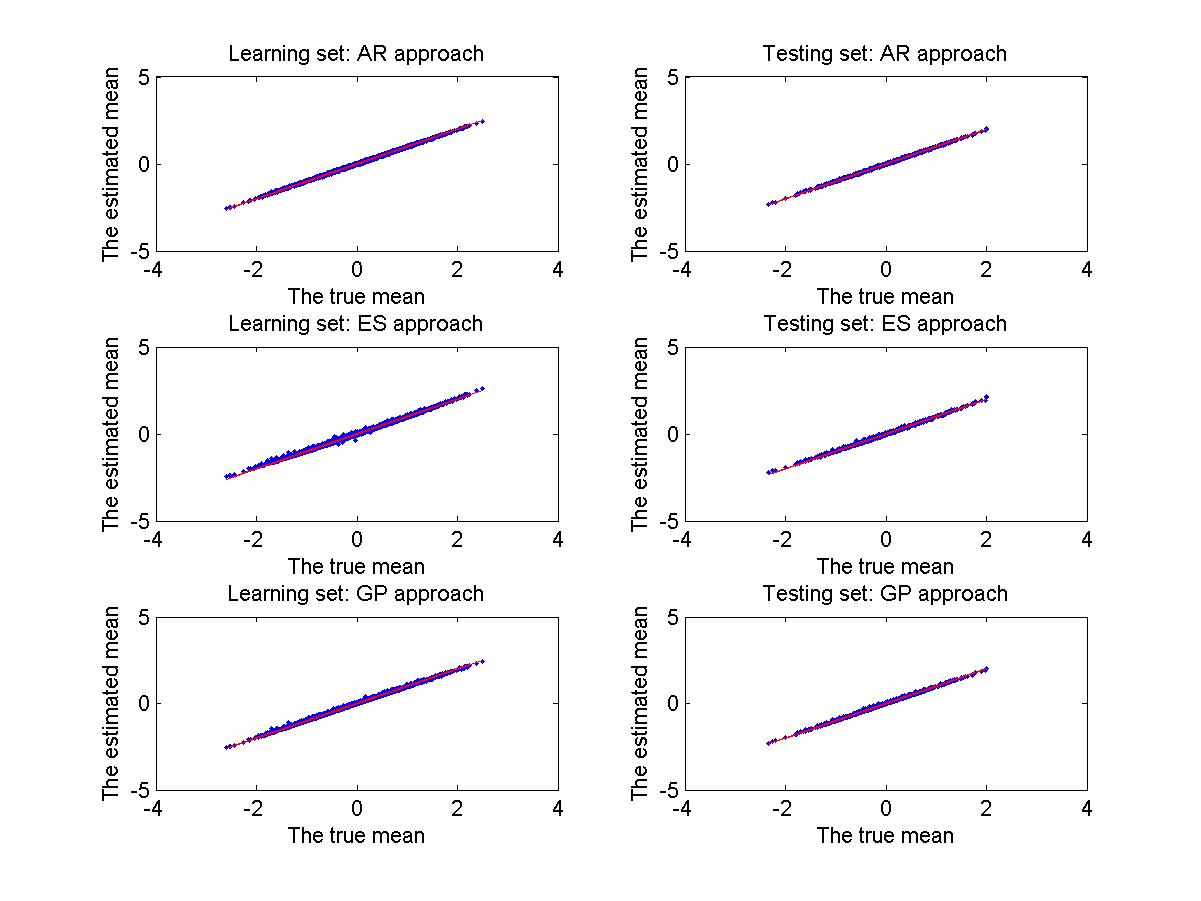}\\
  \caption{The plot of the estimated condition mean of the future return given $\mathcal{F}_{t}$ against its true conditional mean for the data set $r^{(1)}$.}\label{AR_fig}
\end{figure}
\begin{figure}
  \centering
  \includegraphics[width=4 in]{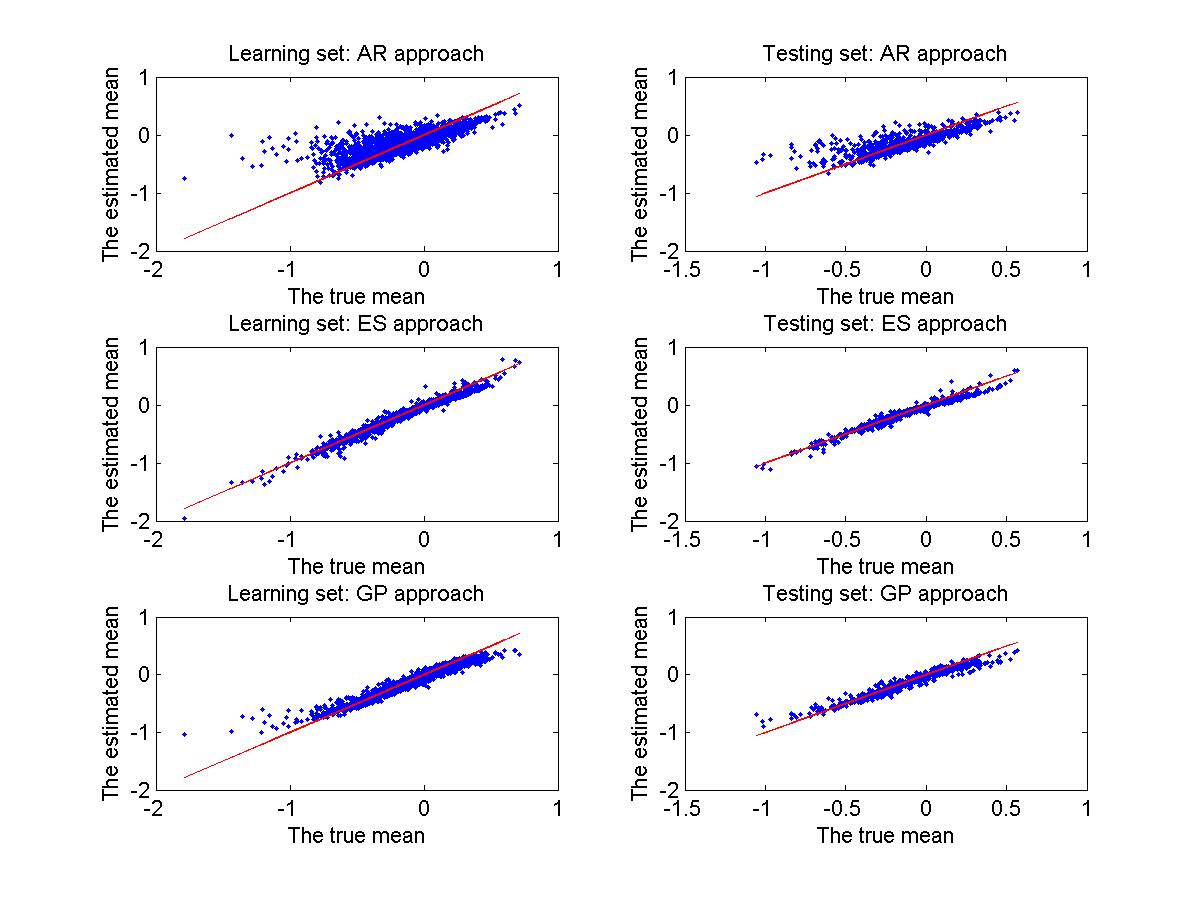}\\
  \caption{The plot of the estimated condition mean of the future return given $\mathcal{F}_{t}$ against its true conditional mean for the data set $r^{(2)}$.}\label{PolyAR}
\end{figure}
\begin{figure}
  \centering
  \includegraphics[width=4 in]{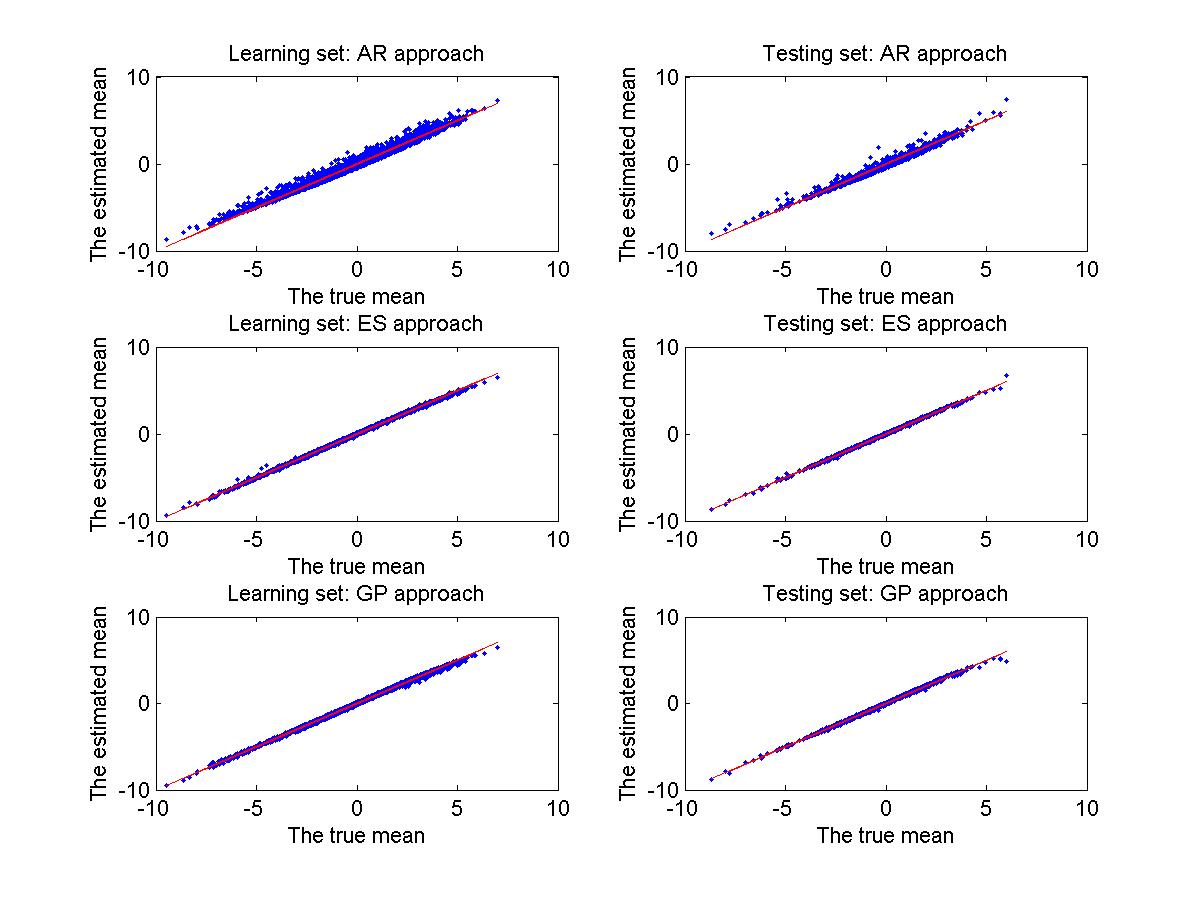}\\
  \caption{The mixture of Poly ARs models}\label{MixPolyAR}
\end{figure}\\
\\
\textbf{Comparison with the AR approach:} Based on Figure \ref{AR_fig}, Figure \ref{PolyAR} and Figure \ref{MixPolyAR}, the AR calibration outperforms the ES approach slightly in terms of MSE for $r^{(1)}$ while the ES approach outperforms the AR approach for $r^{(2)}$ and $r^{(3)}$ significantly.   Using the $R^{2}$ and the adjusted-$R^{2}$ as a measure of the goodness of fitting, the ES approach performs equally well as the AR approach for the data set $r^{(1)}$, while the ES approach outperforms the AR approach for the other two data sets. For example for the data set $r^{(2)}$ the ES approach produces $7.5\%$ more adjusted-$R^{2}$ than the AR approach. The summary of $R^{2}$ and the adjusted-$R^{2}$ statistics are given at Table \ref{R2} and Table \ref{adjusted_R2}. It is expected as the AR calibration is model-specific and it only performs very well when it is a right choice for calibration.\\
\begin{table}
\caption{The summary of the $R^{2}$. }\label{R2}
  \centering
 \begin{tabular}{|c|c|c|c|}
     \hline
       $R^{2}$ &  AR model & Poly AR model & Mix PolyAR model \\
      \hline
      AR& 0.401 & 0.0467 &0.834\\
      \hline
     ES& 0.402& 0.0674 & 0.855\\
     \hline
\end{tabular}
\end{table}
\begin{table}
\caption{The summary of the adjusted-$R^{2}$. }\label{adjusted_R2}
  \centering
 \begin{tabular}{|c|c|c|c|}
     \hline
       Adjusted-$R^{2}$ &   AR model & Poly AR model & Mix PolyAR model \\
      \hline
      AR& 0.4 & 0.0458 & 0.834\\
      \hline
     ES& 0.4& 0.0633 & 0.854 \\
     \hline
\end{tabular}
\end{table}
\\
\textbf{Comparison with the GP approach:} The GP approach produces a similar result to the ES approach in the data sets $r^{(2)}$ and $r^{(3)}$, while the GP approach outperforms the ES approach in the first dataset according to Figure \ref{AR_fig}, Figure \ref{PolyAR} and Figure \ref{MixPolyAR}. However, the computation time of the ES approach is less than $1/32$ of that of the GP approach, which is shown in Table \ref{Time}.
\begin{table}
\caption{The summary of computational time. }\label{Time}
  \centering
 \begin{tabular}{|c|c|c|c|}
     \hline
       Time(s) &  AR model & Poly AR model & Mix PolyAR model\\
      \hline
      ES& 3.779458& 4.430892 &3.387578\\
      \hline
     GP&136.038301 & 139.154203 & 116.702216 \\
     \hline
\end{tabular}
\end{table}\\
\\
\textbf{Cross validation tests:} After investigating a single run of three above methods, we implement cross validation tests to estimate how accurately those predictive models. The results of 20-fold cross validation tests shown in Figure \ref{ARCrossVal}, Figure \ref{PolyARCrossVal} and Figure \ref{MixARCrossVal}, confirm again that the previous fitting result. Those figures along with Table \ref{MSE_summary} suggest that the ES approach and the GP approach are more robust than the AR approach, especially when the model mis-specification occurs. Compared with a non-parametric counterpart - the GP approach, the ES approach generated comparable fitting results, but took much less computational time. We have tested time series with mean equations of different kinds, like periodic function and so on. We find that applying local regression on the signatures can produce the similar results as the GP with less computational time. The numerical examples we showed here are typical representatives , and our finding is valid in general. Based on those numerical evidence, the ES approach is an alternative way with much robustness and efficiency in dealing with the sequential data compared with the standard parametric and non-parametric time series models. 
\begin{figure}
  \centering
  \includegraphics[width=4 in]{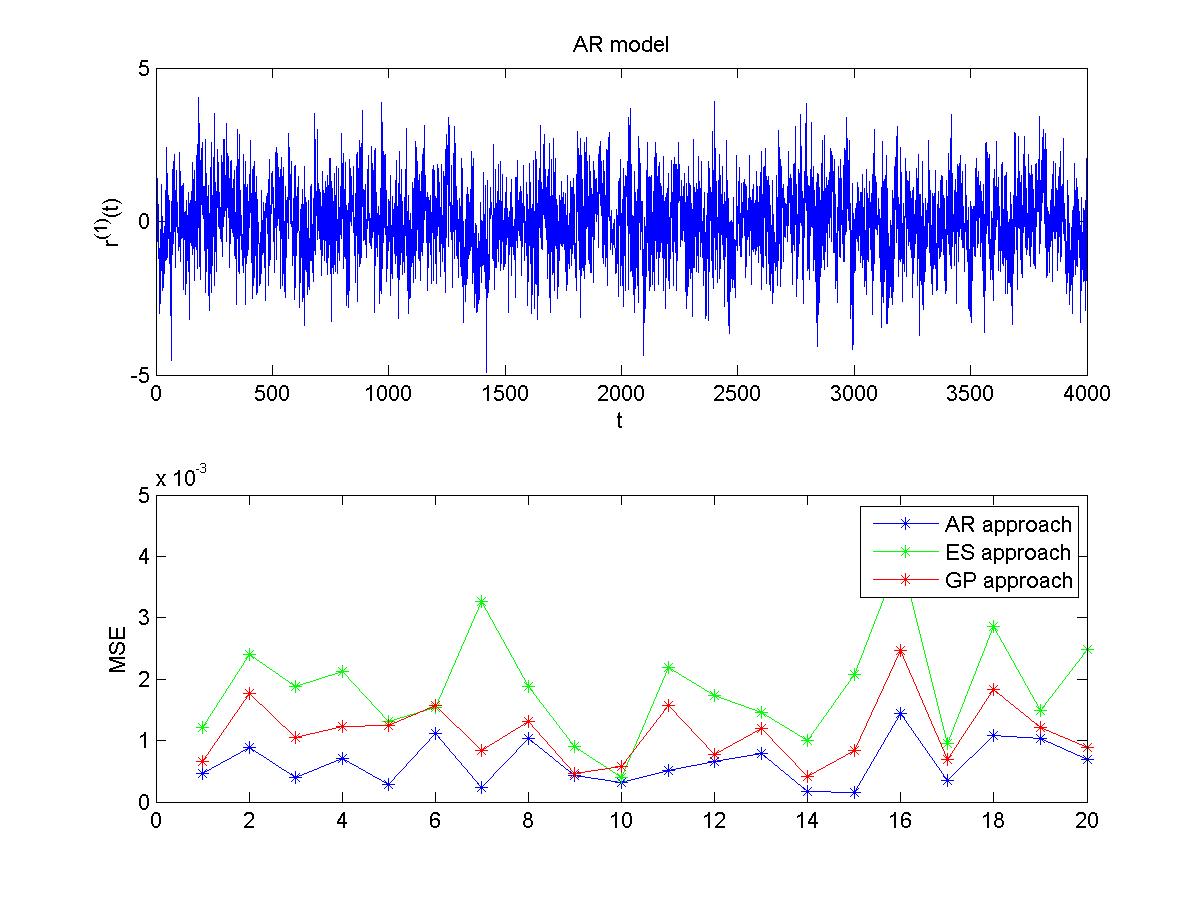}\\
  \caption{The plot of the dataset $r^{(1)}$ and the MSE plot for cross validation tests.}\label{ARCrossVal}
\end{figure}
\begin{figure}
  \centering
  \includegraphics[width=4 in]{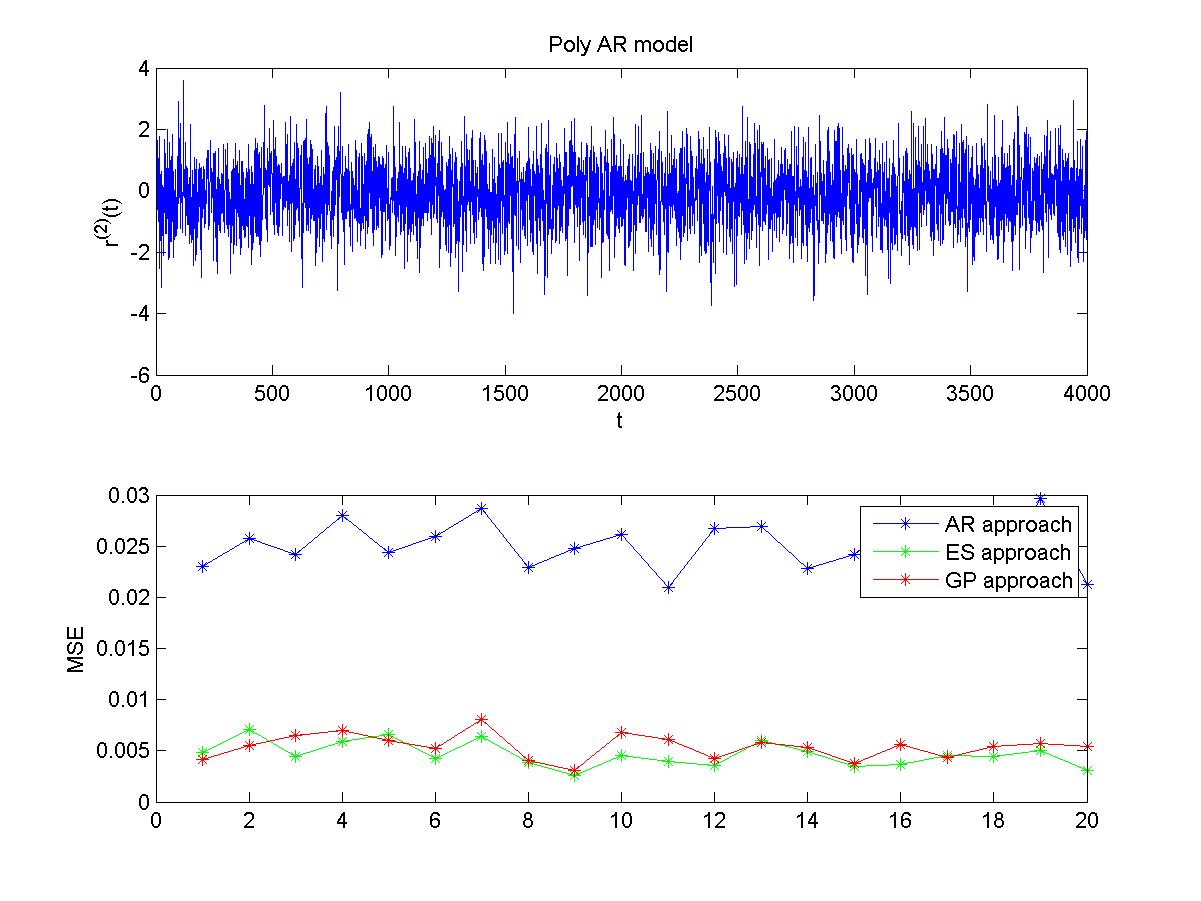}\\
  \caption{The plot of the dataset $r^{(2)}$ and  the MSE plot for cross validation tests.}\label{PolyARCrossVal}
\end{figure}
\begin{figure}
  \centering
  \includegraphics[width=4in]{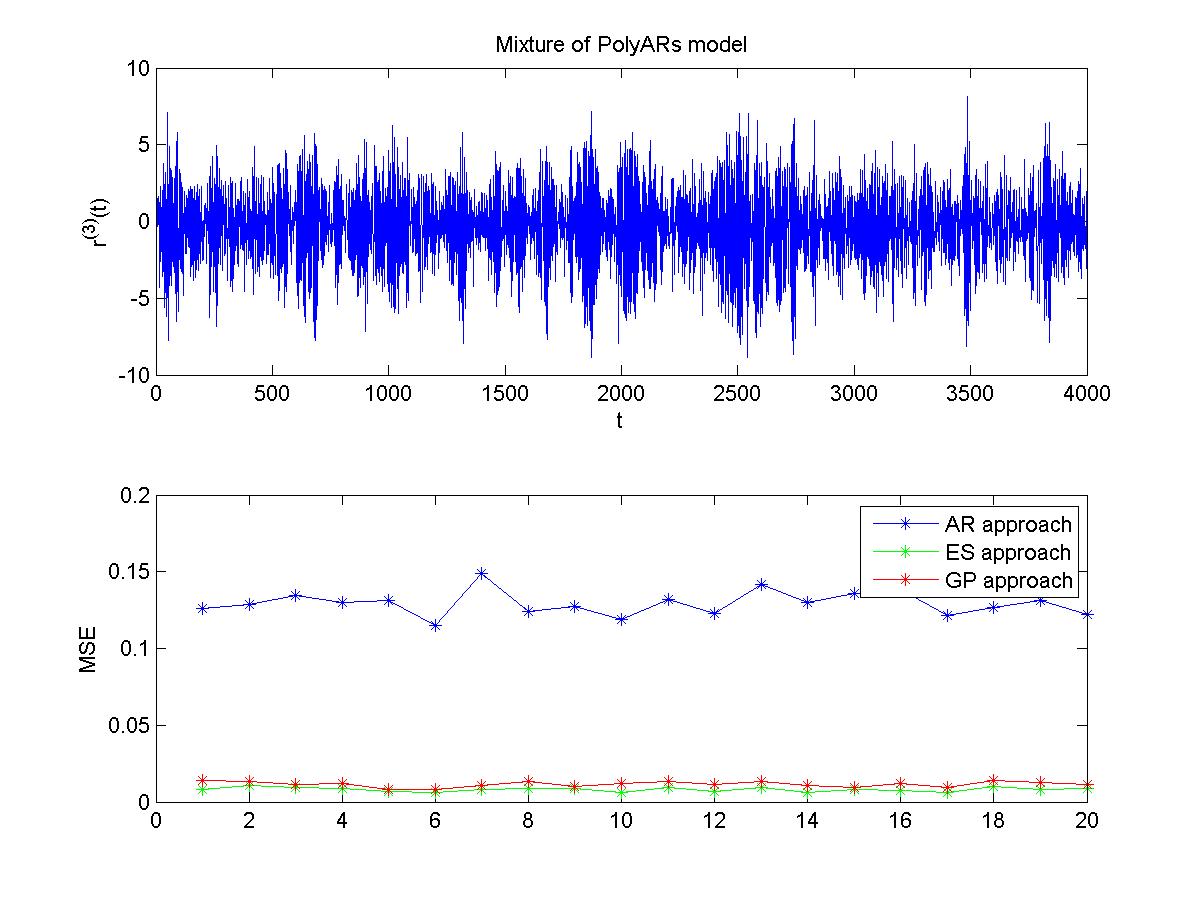}\\
  \caption{The plot of the dataset $r^{(3)}$ and the MSE plot for cross validation tests. }\label{MixARCrossVal}
\end{figure}
\begin{table}
\caption{The summary of $\text{MSE}_{cv}$}\label{MSE_summary}
\begin{center}
\begin{tabular}{|c|c|c|c|}
  \hline
 $\text{MSE}_{cv}$ & AR model & Poly AR model& Mix PolyAR model \\
  \hline
 AR & 6.4e-4 & 0.025 & 0.1294\\
  \hline
  ES & 0.0019& 0.0047& 0.0082 \\
  \hline
  GP & 0.0011  & 0.0054 & 0.0116 \\
  \hline
\end{tabular}\label{Table_MixAR_MSE}
\end{center}
\end{table}

\appendix

\section{Rough Paths}\label{AppendixRoughPath}
\begin{definition}[Rough Path] Let $\mathbf{Z}: \Delta_{T} \rightarrow T^{\lfloor p \rfloor}(E)$ be a continuous map, where $\Delta_{T}:= \{(s, t) \vert 0 \leq s \leq t \leq T\}$. If the following conditions are satisfied:
\begin{enumerate}
\item The functional $\mathbf{Z}$ is a multiplicative functional, i.e. $\pi^{()}(\mathbf{Z}(s, u)) \equiv 1$ and for any $(s, u, t)$ such that $ 0 \leq s \leq u \leq t \leq T$,  it holds that
\begin{eqnarray}\label{Chen2}
\mathbf{Z}(s, u) \otimes \mathbf{Z}(u, t) = \mathbf{Z}(s, t). 
\end{eqnarray}
\item $\mathbf{Z}$ has finite $p$-variation, that is, for every $i = 1, \dots, \lfloor p \rfloor$, and $(s,t) \in \Delta_{T}$, it satisfies
\begin{eqnarray*}
\vert\vert \mathbf{Z}_{s, t}^{i} \vert\vert \leq \frac{(\omega(t -s))^{i/ p}}{\beta (i/p)!}
\end{eqnarray*}
where $\vert \vert . \vert \vert$ is the Euclidean norm is the appropriate dimension, $\omega$ is a control function and $\beta$ is a real number depending only on $p$.
\end{enumerate}
then we say that $\mathbf{Z}$ is a $p$-rough path. The space of $p$-rough paths in $E$ is denoted by $\Omega_{p}(E)$.
\end{definition}
It is noted that Equation (\ref{Chen2}) comes from Chen's identity of signatures.
\begin{definition}[Geometric Rough Path] A geometric $p$-rough path is a $p$-rough path that can be expressed as a limit of $1$-rough path in the $p$-variation distance $d_{p}$, defined as follows: for any $\mathbf{X}, \mathbf{Y}$ continuous functions from $\Delta_{T}$ to $T^{n}(E)$, such that
\begin{eqnarray*}
d_{p}(\mathbf{X}, \mathbf{Y}) = \max_{1 \leq i \leq \lfloor p \rfloor} \sup_{\mathcal{D} \subset [0, T]} \left( \sum_{l} \vert\vert \mathbf{X}_{t_{l-1}, t_{l}}^{i} -  \mathbf{Y}_{t_{l-1}, t_{l}}^{i}\vert\vert^{p / i}\right)^{i/p},
\end{eqnarray*}
where $\mathcal{D} = \{t_{l}\}_{l}$ is taken over all possible finite partition of $[0, T]$. The space of geometric $p$-rough paths in $E$ is denoted by $G\Omega_{p}(E)$.
\end{definition}
The following theorem is the main ingredient of rough paths theory and called "extension theorem", which states that every multiplicative functional of degree $n$ with finite $p$-variation can be extended in a unique way to a multiplicative functional with finite $p$-variation of arbitrary high degree, provided $n \geq \lfloor p \rfloor$ .
\begin{theorem}[Theorem 3.7, \cite{RoughPaths}] Let $p \geq 1$ be a real number and $k \geq 1$ be an integer. Let $\mathbf{X}: \Delta_{T} \rightarrow T^{k}(E)$ be a multiplicative functional with finite $p$-variation. Assume that $k \geq \lfloor p \rfloor$. Then there exists a unique extension of $\mathbf{X}$ to a multiplicative functional $\hat{X}: \Delta_{T} \rightarrow T^{k+1}(E)$ of finite $p$-variation.
\end{theorem}
\section{Time series models}
\begin{definition}[AR model]\label{AR}
Let $\{r_{t}\}$ be a time series. The notation $AR(p)$ indicates an autoregressive model of order $p$. The $AR(p)$ model with parameters $\Phi = [\Phi_{0}, \dots, \Phi_{p}]$ is defined as follows:
\begin{eqnarray*}
r_{t} = \Phi_{0} + \Phi_{1}r_{t-1}+ \dots + \Phi_{p}r_{t-p} + a_{t},
\end{eqnarray*}
where $p$ is a non-negative integer and $a_{t}$ is a white noise with mean zero and variance $\sigma_{a}^{2}$.
\end{definition}
\begin{definition}[ARCH model]\label{Def_ARCH}
We say that a time series $\{r_{k}\}$ satisfies the assumptions of $ARCH(q)$ model, if and only if the error terms $\varepsilon_{k}$ defined by return residuals with respect to a mean process, (i.e. $\varepsilon_{k} = r_{k} - \mu_{k}$) satisfy the following equations:
\begin{eqnarray*}
\varepsilon_{k} &=& \sigma_{k} z_{k};\\
\sigma_{k}^{2} &=& \alpha_{0} + \sum_{i = 1}^{q} \alpha_{i} \varepsilon_{k-i}^{2},
\end{eqnarray*}
where $z_{k}$ is a strong white noise, $\{a_{i}\}_{i =0}^{q}$ are all constants such that $a_{0} > 0$ and $a_{i} \geq 0, \forall i = 1, \cdots, q$ and the mean equation $\mu_{t}$ is a linear combination of the lagged returns, i.e. there exists a positive integer $Q$ and constants $\{\beta_{i}\}_{i=1}^{Q}$, such that
\begin{eqnarray*}
\mu_{t} = \beta_{0}+\sum_{i=1}^{Q}\beta_{i}r_{t-i}.
\end{eqnarray*}
\end{definition}
\section{Numerical examples}\label{AppendixParameter}
\subsection{The Poly AR model}
For the second data set, the AR calibration gives the estimator for $\Phi$ as follows:
\begin{eqnarray*}
\hat{\Phi} := [\hat{\Phi}_{0}, \hat{\Phi}_{1}, \hat{\Phi}_{2}, \hat{\Phi}_{3}] =[-0.11621, 0.018418, -0.023661, 0.21904].
\end{eqnarray*}
We specify $n = 4$ and give all possible indices $I$ and the corresponding estimated coefficients $\{\hat{f}^{I}\}$ as follows:
\begin{eqnarray*}\small
\begin{array}{ccccccc}
\hline
I & ()  & (1) & (2) & (1, 1) &(1,2) & (2,1) \\
\hat{f}^{I} &0&-0.019034&0& 0& 0 & 0 \\
\hline
I &(2,2)&(1,1,1) &(1,1,2) &(1,2,1) &(1,2,2)& (2,1,1)\\
\hat{f}^{I} & 0& 0 &0.11788&0.10232&0&0.11755\\
\hline
I & (2,1,2)& (2,2,1)&(2,2,2)& (1,1,1,1)& (1,1,1,2) &(1,1,2,1)\\
\hat{f}^{I}& -0.015945& 0 & 0 &0 & 0 &0\\
\hline
I & (1,1,2,2) &(1,2,1,1)& (1,2,1,2) &(1,2,2,1)& (1,2,2,2) &(2,1,1,1) \\
\hat{f}^{I}&  -0.017901& 0& 0.071969 &-0.033093& -0.014362 & 0\\
\hline
I& (2,1,1,2)&(2,1,2,1)& (2,1,2,2) &(2,2,1,1) &(2,2,1,2)&(2,2,2,1)\\
\hat{f}^{I} & 0&-0.022093& -0.0061961& -0.0076586& - 0.0099546& -0.01357\\
\hline
I&  (2,2,2,2)&&&& &\\
\hat{f}^{I}& -0.020399&&&& &\\
\hline
\end{array}
\end{eqnarray*}
The GP approach gave fitted parameters $\ln(\lambda) = 1.6529, \ln(h) = 0.0045$.
\subsection{The Mixture of Poly ARs model}
For the third data set, the AR calibration gives the estimator for $\Phi$ as follows:
\begin{eqnarray*}
\hat{\Phi} := [\hat{\Phi}_{0}, \hat{\Phi}_{1}, \hat{\Phi}_{2}, \hat{\Phi}_{3}] =[-0.4909, -0.57929, -0.13741, 0.6482].
\end{eqnarray*}
Let $n = 4$. The inferred parameters based on the ES approach are given as follows:
\begin{eqnarray*}\small
\begin{array}{ccccccc}
\hline
I & ()  & (1) & (2) & (1, 1) &(1,2) & (2,1) \\
\hat{f}^{I} &-0.1655&0&  -0.16315 & 0& 0 & 0 \\
\hline
I &(2,2)&(1,1,1) &(1,1,2) &(1,2,1) &(1,2,2)& (2,1,1)\\
\hat{f}^{I}& -0.25368 &0&0&0&-0.089823& 0.66242\\
\hline
I & (2,1,2)& (2,2,1)&(2,2,2)& (1,1,1,1)& (1,1,1,2) &(1,1,2,1)\\
\hat{f}^{I}& 0 & 0 &0.0364850 & 0 &0&  0\\
\hline
I & (1,1,2,2) &(1,2,1,1)& (1,2,1,2) &(1,2,2,1)& (1,2,2,2) &(2,1,1,1) \\
\hat{f}^{I}& 0&0.35718& -0.029195& -0.001156& -0.057973 & 0\\
\hline
I& (2,1,1,2)&(2,1,2,1)& (2,1,2,2) &(2,2,1,1) &(2,2,1,2)&(2,2,2,1)\\
\hat{f}^{I}& 0& 0.017382& -0.091709&  0.0043497&0&0.01107\\
\hline
I&  (2,2,2,2)&&&& &\\
\hat{f}^{I}& 0.2019&&&& &\\
\hline
\end{array}
\end{eqnarray*}
The GP approach gave the fitted parameters $\ln(\lambda) = 2.8852, \ln(h) = 1.9846$.
\section*{Acknowledgements}
The second author and the third author would like to thank Remy Cottet (AHL), Anthony Ledford (AHL) and Syed Ali Asad Rizvi (the Oxford-Man Institute of Quantitative Finance) for their valuable suggestions. All the authors would like to thank Oxford-Man Institute for their funding.
%\begin{supplement}
%\sname{Supplement A}\label{suppA}
%\stitle{Title of the Supplement A}
%\slink[url]{http://www.e-publications.org/ims/support/dowload/imsart-ims.zip}
%\sdescription{Dum esset rex in
%accubitu suo, nardus mea dedit odorem suavitatis. Quoniam confortavit
%seras portarum tuarum, benedixit filiis tuis in te. Qui posuit fines tuos}
%\end{supplement}

\bibliographystyle{plain}
\bibliography{myref}

\begin{thebibliography}{10}

\bibitem{billings2005new}
Stephen Billings, Hua-Liang Wei, et~al.
\newblock A new class of wavelet networks for nonlinear system identification.
\newblock {\em Neural Networks, IEEE Transactions on}, 16(4):862--874, 2005.

\bibitem{boedihardjo2014signature}
H~Boedihardjo, X~Geng, T~Lyons, and D~Yang.
\newblock The signature of a rough path: Uniqueness.
\newblock {\em Advances in Mathematics, arXiv:1406.7871}, pages 720--737, 2016.

\bibitem{boedihardjo2014uniqueness}
Horatio Boedihardjo and Xi~Geng.
\newblock The uniqueness of signature problem in the non-markov setting.
\newblock {\em Stochastic Processes and their Applications},
  125(12):4674--4701, 2015.

\bibitem{boedihardjo2015uniform}
Horatio Boedihardjo, Terry Lyons, and Danyu Yang.
\newblock Uniform factorial decay estimates for controlled differential
  equations.
\newblock {\em Electronic Communications in Probability}, 20:1--11, 2015.

\bibitem{bollerslev1990modelling}
Tim Bollerslev.
\newblock Modelling the coherence in short-run nominal exchange rates: a
  multivariate generalized arch model.
\newblock {\em The Review of Economics and Statistics}, pages 498--505, 1990.

\bibitem{chen1957integration}
Kuo-Tsai Chen.
\newblock Integration of paths, geometric invariants and a generalized
  baker-hausdorff formula.
\newblock {\em Annals of Mathematics}, pages 163--178, 1957.

\bibitem{chen1977iterated}
Kuo-Tsai Chen.
\newblock Iterated path integrals.
\newblock {\em Bulletin of the American Mathematical Society}, 83(5):831--879,
  1977.

\bibitem{chevyrev2013set}
Ilya Chevyrev and Terry Lyons.
\newblock Characteristic functions of measures on geometric rough paths.
\newblock {\em Annals of Probability, accepted, arXiv: 1307.3580v5}, 2015.

\bibitem{clark1980maximum}
J~Clark and R~Cameron.
\newblock The maximum rate of convergence of discrete approximations for
  stochastic differential equations.
\newblock {\em Stochastic Differential Systems Filtering and Control}, pages
  162--171, 1980.

\bibitem{dickinson2007optimal}
Andrew~S Dickinson.
\newblock Optimal approximation of the second iterated integral of brownian
  motion.
\newblock {\em Stochastic Analysis and Applications}, 25(5):1109--1128, 2007.

\bibitem{Fawcett}
T.~Fawcett.
\newblock {\em Problems in stochastic analysis: connections between rough paths
  and noncommutative harmonic analysis}.
\newblock PhD thesis, University of Oxford, 2003.

\bibitem{flint2013discretely}
Guy Flint, Ben Hambly, and Terry Lyons.
\newblock Discretely sampled signals and the rough hoff process.
\newblock {\em Stochastic Processes. Appl., accepted, arXiv:1310.4054}, 2016.

\bibitem{friz2015physical}
Peter Friz, Paul Gassiat, and Terry Lyons.
\newblock Physical brownian motion in a magnetic field as a rough path.
\newblock {\em Transactions of the American Mathematical Society}, 2015.

\bibitem{friz2008euler}
Peter Friz and Nicolas Victoir.
\newblock Euler estimates for rough differential equations.
\newblock {\em Journal of Differential Equations}, 244(2):388--412, 2008.

\bibitem{friz2014course}
Peter~K Friz and Martin Hairer.
\newblock {\em A Course on Rough Paths: With an Introduction to Regularity
  Structures}.
\newblock Springer, 2014.

\bibitem{Greg}
Lajos~Gergely Gyurk{\'o}, Terry Lyons, Mark Kontkowski, and Jonathan Field.
\newblock Extracting information from the signature of a financial data stream.
\newblock {\em arXiv preprint arXiv:1307.7244}, 2013.

\bibitem{UniquenessOfSignature}
B.M. Hambly and Terry Lyons.
\newblock Uniqueness for the signature of a path of bounded variation and the
  reduced path group.
\newblock {\em Annals of Mathematics}, 171(1):109--167, 2010.

\bibitem{hardle1990applied}
Wolfgang Hardle.
\newblock {\em Applied nonparametric regression}, volume~27.
\newblock Cambridge Univ Press, 1990.

\bibitem{le2013stratonovich}
Yves Le~Jan and Zhongmin Qian.
\newblock Stratonovich’s signatures of brownian motion determine brownian
  sample paths.
\newblock {\em Probability Theory and Related Fields}, 157(1-2):209--223, 2013.

\bibitem{RoughPaths}
Terry Lyons, Thierry L$\acute{e}$vy, and Michael Caruana.
\newblock {\em Differential Equation driven by Rough Paths}.
\newblock Springer, 2006.

\bibitem{lyons2014feature}
Terry Lyons, Hao Ni, and Harald Oberhauser.
\newblock A feature set for streams and an application to high-frequency
  financial tick data.
\newblock In {\em Proceedings of the 2014 International Conference on Big Data
  Science and Computing}, page~5. ACM, 2014.

\bibitem{lyons2002system}
Terry Lyons and Zhongmin Qian.
\newblock {\em System control and rough paths}.
\newblock Oxford University Press, 2002.

\bibitem{lyons1998differential}
Terry~J Lyons.
\newblock Differential equations driven by rough signals.
\newblock {\em Revista Matem{\'a}tica Iberoamericana}, 14(2):215--310, 1998.

\bibitem{papavasiliou2011parameter}
Anastasia Papavasiliou, Christophe Ladroue, et~al.
\newblock Parameter estimation for rough differential equations.
\newblock {\em The Annals of Statistics}, 39(4):2047--2073, 2011.

\bibitem{rasmussen2006gaussian}
C.~E. Rasmussen and C.~K.~I. Williams.
\newblock {\em Gaussian Processes for Machine Learning}.
\newblock the MIT Press, 2006.

\bibitem{AnalysisofFinancialTimeSeries}
Ruey~S. Tsay.
\newblock {\em Analysis of Financial Time Series}.
\newblock Wiley series in probability and statistics. Wiley, 2 edition.

\end{thebibliography}

\end{document}